\documentclass[acmsmall,nonacm]{acmart}
\AtBeginDocument{%
  }

\acmISBN{978-1-4503-XXXX-X/2018/06}

\newcommand{\reduce}{\textup{\textsf{reduce}}}
\newcommand{\union}{\textup{\textsf{union}}}
\newcommand{\estimateAndSample}{\textup{\textsf{estimateAndSample}}}
\newcommand{\approxMCNFAcore}{\textup{\textsf{countNFAcore}}}
\newcommand{\approxMCNFA}{\textup{\textsf{countNFA}}}
\newcommand{\computeCache}[1]{\textup{\textsf{computeCache}}(#1)}
\newcommand{\updateCache}[1]{\textup{\textsf{updateCache}}(#1)}

\newcommand{\pred}{\mathsf{pred}}
\newcommand{\ancestors}{\mathsf{ancestors}}
\newcommand{\Ex}{\textsf{E}}
\newcommand{\Va}{\textsf{Var}}
\renewcommand{\min}{\textsf{min}}

\newcommand{\calH}{\mathcal{H}}

\newcommand{\calA}{\mathcal{A}}
\newcommand{\calQ}{\mathcal{Q}}
\newcommand{\calT}{\mathcal{T}}
\newcommand{\calL}{\mathcal{L}}
\newcommand{\calS}{\mathcal{S}}

\newcommand{\cache}{\mathsf{cache}}

\newcommand{\transition}{\mathsf{transition}}
\newcommand{\paths}{\mathsf{run}}
\newcommand{\lcpn}{\mathsf{lcps}}
\newcommand{\median}{\mathsf{median}}
\newcommand{\vY}[4]{\mathfrak{S}^{#1}_{#2,#3}(#4)}
\newcommand{\vZ}[3]{\hat{\mathfrak{S}}^{#1}_{#2}(#3)}

\newcommand{\ns}{n_s}
\newcommand{\nt}{n_t}

\usepackage{thm-restate}  
  
\usepackage[linesnumbered,lined,ruled,vlined]{algorithm2e}
\usepackage{tikz}
\usetikzlibrary{decorations.pathreplacing,tikzmark,calc,calligraphy}
\usepackage{bbm}

\usepackage{subcaption}
\usepackage{todonotes}
\newtheorem{fact}{Fact}
\newtheorem{claim}{Claim}

\usepackage{mathtools}
\usepackage{enumitem}

\usepackage{algorithmic}

\begin{document}

\title{Towards practical FPRAS for \#NFA: Exploiting the Power of Dependence}
\titlenote{All authors contributed equally to this research. The symbol \textcircled{r} denotes random author order. The publicly
	verifiable record of the randomization is available at \protect\url{https://www.aeaweb.org/journals/policies/random-author-order}}

\author{Kuldeep S. Meel} %
\affiliation{%
\institution{University of Toronto and Georgia Institute of Technology}
\city{Toronto}
\country{Canada}}
\email{meel@cs.toronto.edu}
\orcid{0000-0001-9423-5270}

\author{Alexis de Colnet} %
\affiliation{%
\institution{TU Wien}
\city{Vienna}\country{Austria}}
\email{decolnet@ac.tuwien.ac.at}
\orcid{0000-0002-7517-6735}

 \renewcommand{\shortauthors}{Kuldeep S. Meel \& Alexis de Colnet}

\begin{abstract}
\#NFA refers to the problem of counting the words of length $n$ accepted by a non-deterministic finite automaton. \#NFA is \#P-hard, and although fully-polynomial-time randomized approximation schemes (FPRAS) exist, they are all impractical. The first FPRAS for \#NFA had a running time of $\tilde{O}(n^{17}m^{17}\varepsilon^{-14}\log(\delta^{-1}))$, where $m$ is the number of states in the automaton, $\delta \in (0,1]$ is the confidence parameter, and $\varepsilon > 0$ is the tolerance parameter (typically smaller than $1$). The current best FPRAS achieved a significant improvement in the time complexity relative to the first FPRAS and obtained FPRAS with time complexity $\tilde{O}((n^{10}m^2 + n^6m^3)\varepsilon^{-4}\log^2(\delta^{-1}))$. The complexity of the improved FPRAS is still too intimidating to attempt any practical implementation. 
	
	In this paper, we pursue the quest for practical FPRAS for \#NFA by presenting a new  algorithm with a time complexity of $O(n^2m^3\log(nm)\varepsilon^{-2}\log(\delta^{-1}))$. Observe that evaluating whether a word of length $n$ is accepted by an NFA has a time complexity of $O(nm^2)$. Therefore, our proposed FPRAS achieves sub-quadratic complexity with respect to membership checks.
\end{abstract}

\ccsdesc[500]{Theory of computation~Regular languages}
\ccsdesc[500]{Theory of computation~Database query processing and optimization (theory)}
\ccsdesc[500]{Mathematics of computing~Probabilistic algorithms}
\ccsdesc[500]{Mathematics of computing~Stochastic processes}
\keywords{Non-deterministic Finite Automaton, NFA, Model Counting, FPRAS}
\maketitle

\section{Introduction}\label{section:introduction}

In this paper, we focus on the following computational problem:

\begin{description}
	\item[\quad\quad{\bf\#NFA}:] Given a non-deterministic finite automaton (NFA) $\calA = (\calQ,q_I,q_F,\calT)$ over the binary alphabet $\Sigma = \{0,1\}$ with $m$ states, and a number $n$ in unary, determine $|\calL_n(\calA)|$ the number of words of length $n$ accepted by $\calA$.
\end{description}

The notion of finite automaton is one of the fundamental notions in Computer Science, and therefore, serves as fundamental object for representation of computational processes. 
We briefly review some of the applications of \#NFA in the context of databases community. The description of these applications is largely borrowed from earlier works on \#NFA~\cite{MeelCM24}.

\paragraph{Probabilistic Query Evaluation}

Given a query $Q$ and a probabilistic database $D$, probabilistic query evaluation (PQE) is the problem of computing the probability that $Q$ holds true on a randomly sampled instance of $D$. It is known that for specific query classes, it can be efficiently reduced to the problem of determining the number of paths in an NFA. Notably, for self-join-free path queries over binary relations, the reduction is linear in the size of both $Q$ and $D$~\cite{vanBremen2023}. Therefore, a practical algorithm for \#NFA has immediate implications in the context of probabilistic query evaluation. 

\paragraph{Counting Answers to Regular Path Queries}

Regular path queries (RPQs) are a fundamental construct in graph query languages. They allow us to count the paths between nodes that match a given regular expression. This problem can be efficiently reduced to counting paths in an NFA. The resulting NFA is constructed by combining the NFAs representing the database and the regular expression. The reduction is linear in the size of both the query and the database~\cite{ACJR19}.  

\paragraph{Probabilistic Graph Homomorphisms}

A probabilistic graph $(H, \pi)$ is a graph $H$ where each edge $e$ is associated with a probability $\pi(e)$. The probabilistic graph homomorphism problem asks for the probability that a random subgraph of $H$ contains a specific pattern (given by a query graph $G$). For  path queries such as 1-way path queries, this problem can be reduced to  \#NFA~\cite{AvBM23}.

Given the widespread usage of automaton as a modeling tool for computational processes, \#NFA finds applications in  verification~\cite{SMC2010}, 
testing of software~\cite{sutton2007fuzzing}, distributed systems~\cite{Ozkan2019}, 
quantitative information flow~\cite{Bang2016,Gao2019,Saha2023},
information extraction~\cite{ACJR19}, and music generation~\cite{Donze2014}.
Given such applications, it is highly desirable to have a practical algorithm for \#NFA.

From a theoretical perspective, \#NFA is a \#P-complete problem under Turing reduction and it is also shown complete for the class span-L under parsimonious reduction~\cite{AlvarezJ93}. The membership in span-L indicated possibility of polynomial-time randomized approximation, and consequently, finding such an algorithm became a major open problem. The first major advance was achieved by Kannan, Sweedyk, and Mahaney~\cite{KSM95} who proposed the first quasi-polynomial randomized approximation scheme for \#NFA. In a major breakthrough, Arenas, Croquevielle, Jayaram, and Riveros (henceforth referred as ACJR, after initials of the authors) presented the first fully polynomial randomized approximation scheme (FPRAS)~\cite{ACJR19}. As a corollary, ACJR's result also implied that every problem in span-L admits FPRAS.

The discovery of FPRAS did not lend itself to practical implementations as the time complexity of ACJR's scheme is too high for practical implementation.  In particular, the scheme due to ACJR had time complexity of  $\tilde{O}(n^{17}m^{17}\varepsilon^{-14}\log(\delta^{-1}))$, where $m$ is the number of states of the automaton, $n$ is the word length, and $\delta \in (0,1]$ and $\varepsilon > 0$ are parameters of the FPRAS. Meel, Chakraborty, and Mathur (MCM)~\cite{MeelCM24} managed to achieve a significant improvement in runtime by achieving an FPRAS scheme with time complexity $\tilde{O}((n^{10}m^2 + n^6m^3)\varepsilon^{-4}\log^2(\delta^{-1}))$. While the time complexity of MCM's scheme is a significant improvement compared to ACJR's scheme, it is still hopelessly impractical. Therefore, the major challenge remains: Can we design an FPRAS that is amenable to practical implementation?

The primary contribution of this work is to design a new FPRAS with time complexity of  $O(n^2m^3\log(nm)\varepsilon^{-2}\log(\delta^{-1}))$, which achieves significant improvement for all parameters, i.e., $n,m,\varepsilon$ and $\delta$.   It is perhaps worth emphasizing that membership check for NFA has the complexity of $O(nm^2)$, therefore the complexity of our FPRAS is sub-quadratic in the time complexity of membership check.  The contribution is formalized in the following theorem:
\begin{restatable}{theorem}{mainResult}\label{theorem:main_result}
	For $\calA = (\calQ,q_I,q_F,\calT)$ an NFA with $m =|\calQ|$, $n$ a positive integer in unary, $\varepsilon > 0$ and $\delta > 0$, 
	$\approxMCNFA(\calA,n,\varepsilon,\delta)$ runs in time $O(n^2m^3\log(nm)\varepsilon^{-2}\log(\delta^{-1}))$ and returns $\mathsf{est}$ with the guarantee that
	$
	\Pr\left[\mathsf{est}  \in (1 \pm \varepsilon) |\calL_n(\calA)| \right] \geq 1 - \delta.
	$
\end{restatable}

Since we are interested in designing a practical FPRAS, we also complement our description of the algorithm with detailed implementation choices, in particular usage of different caching mechanisms that may be more amenable to practical implementation over GPUs that enable scalable matrix operations.  

 Achieving such a significant reduction in the time complexity required a completely different approach in comparison to prior work. Prior work on the design of FPRAS pioneered a breakthrough idea of reusing the information from samples in the previous layers but not to reuse the samples themselves as reuse would introduce dependence. In this work, we give up the desire for independence and instead reuse samples. The reuse of samples necessitates a different analysis and to this end, our analysis follows the new technique proposed by Meel and de Colnet~\cite{MdC24a} in the context of the design of FPRAS for \#nFBDD. Meel and de Colnet showed how dependence among samples can be analyzed by careful analysis of derivation paths and coupling via random process; we adapt such ideas in the context of NFA. It is perhaps worth remarking that simply appealing to the result of~\cite{MdC24a} would lead to a running time complexity of $O(n^5m^6 \varepsilon^{-4} \log \delta^{-1})$\footnote{A careful reader will note the time complexity expressions stated here is better than the expression stated in Theorem~1 of~\cite{MdC24a}. If we were to simply appeal to Theorem 1 where we replace $|B|$ by $O(nm^2)$, the time complexity would be $O(n^{11}m^{12} \varepsilon^{-4} \log \delta^{-1})$. The improved expression is due to the observation that the unrolled automaton is already a {\em 1-complete BDD}}.

 While our algorithm achieves a time complexity that is sub-quadratic in the time complexity for membership check, there is still work to be done in realizing practical implementation. A natural next step for future work would be the design of efficient data structures and algorithmic engineering to achieve a practical implementation of the proposed FPRAS.

\subsection{Technical Overview}

We begin with a discussion of similarities of our scheme with that of prior work and then discuss the technical differences that have allowed us to achieve significant improvement in the runtime.

\begin{algorithm}
	\caption{FPRAS Template for \#NFA}	\label{alg:fprastemplate}
	\KwIn{NFA $\calA = (\calQ,q_I,q_F,\calT)$ with $\calQ$ a set of $m$ states, a transition relation $\calT$, a single initial state $q_I \in \calQ$, and a single final state $q_F \in \calQ$. A number $n \in \mathbb{N}$ (in unary).}

	\BlankLine
	Unroll the automaton $\calA$ into an acyclic graph $\calA^{\mathsf{u}}_n$  by making $n+1$ copies $\{q^\ell\}_{l \in [0,n],q \in \calQ}$ of the states $\calQ$ and adding transitions between immediate layers.
	
	\For{$\ell \in \{0, \dots, n\}$ \text{and} $q \in \calQ$}
	{
		Compute the estimate $N(q^{\ell})$ using prior estimates and samples $\{N(\hat{q}^i), S(\hat{q}^i)\}_{i<\ell, \hat{q} \in \calQ}$\;
		Compute a uniform set of samples $S(q^\ell)$ using $N(q^\ell)$ and the prior estimates and samples $\{N(\hat{q}^i), S(\hat{q}^i)\}_{i<\ell, \hat{q} \in \calQ}$\;
	}
	\Return $N(q_F^n)$
\end{algorithm}

Similar to previous FPRAS schemes for \#NFA, our algorithm exploits the tight connection between counting and sampling, and works in a {\em bottom-up} manner. In this context, all the FPRAS schemes can be unified to fit the template framework (also presented in~\cite{ACJR19,MeelCM24}) presented in Algorithm~\ref{alg:fprastemplate}. For technical convenience, we work with an unrolled automaton, which is an acyclic graph with $n+1$
levels; each level $\ell$ containing the $\ell^\text{th}$ copy $q^{\ell}$
of each state $q$ of $\calA$ (as is standard in the literature, we can assume that $\calA$ has one accepting state without loss of generality). On a high-level, all the schemes work in a {\em bottom-up} manner by keeping track of two quantities:

\begin{enumerate}
\item An  estimate of $|\calL(q^{\ell})|$, where $\calL(q^{\ell})$ is the language of length-$\ell$ words whose run ends in $q$. To retain consistency with prior work, let us call the estimate as $N(q^{\ell})$
\item A set of samples $S(q^{\ell})$ that consist of words belonging to $\calL(q^{\ell})$
\end{enumerate}

Before we discuss the mechanisms for how $N(q^{\ell})$ and $S(q^{\ell})$ are constructed, it is instructive to observe the recursive formulation for $\calL(q^{\ell})$. 
 \begin{align*}
	\mathcal{L}(q^\ell) = \bigg( \bigcup_{(q_0,0,q) \in \calT} \mathcal{L}(q_0^{\ell-1}) \bigg) \cdot \{0\} \quad \cup \quad \bigg( \bigcup_{(q_1,1,q) \in \calT} \mathcal{L}(q_1^{\ell-1}) \bigg) \cdot \{1\}
\end{align*}

The primary difficulty lies in the overlapping of different terms in the union expressions. This is where the previous schemes appealed to Monte Carlo-based classical union of sets estimation techniques, which would ideally operate on samples that were sampled independently and uniformly at random. It is easy to see that if we were to construct $S(q^{\ell})$ by reusing samples from $S(q_0^{\ell-1})$ and $S(q_1^{\ell-1})$, then this would introduce dependence among the sets $S(q^{\ell})$ for different $q$. To avoid that, the core insight in~\cite{ACJR19,ACJR21} was to rely on the \emph{self-reducibility} union property, which states that the following is true: for every word $w$ of length $k$, all words from $\calL(q^{\ell})$ with suffix $w$ can be expressed as $\bigcup_{\hat q \in P} \calL(\hat q^{\ell-k})\cdot\{w\}$ for some set $P \subseteq \calQ$ depending only on $w$, $q$ and $k$. The self-reducibility property allows us to sample words, character by character, by growing suffix. For a suffix $w$ of length $k$, one has to decide whether to extend it (from the left) into $0 w$ or into $1 w$. By self-reducibility, the subsets of $\calL(q^\ell)$ with suffix $0 w$ or $1 w$  can be expressed as $W_0 = \bigcup_{\hat q \in P_0} \calL(\hat q^{\ell-k-1})$ and $W_1 = \bigcup_{\hat q \in P_1} \calL(\hat q^{\ell-k-1})$ for some easily computable $P_0$ and $P_1$. Since $W_0$ and $W_1$ are unions of sets (the $\calL(\hat q^{\ell-k-1})$'s) for which we already have size estimates (the $N(\hat q^{\ell-k-1})$'s) and samples (the $S(\hat q^{\ell-k-1})$'s), the Monte Carlo method is used to compute estimates $N(W_0)$ and $N(W_1)$ on their sizes and we grow $w$ into $b w$ with probability $N(W_b)/(N(W_0)+N(W_1))$ (in first approximation). In this sampling scheme, words in $S(q^\ell)$ come from an empty suffix that is grown $\ell$ times; the previous samples in $S(\hat q^i)$ ($i < \ell$) are used in the computation but are not extended themselves, i.e., it is not the case that $S(q^\ell)$ is simply constructed by extending words from some set $S(\hat q^{\ell-1})$. Paraphrasing~\cite{MeelCM24}, the reuse of samples is avoided as it introduces dependence among sets $S(q^{\ell})$ for different $q$.

{\em The core technical insight of our work is to go against  the conventional wisdom espoused in prior work and reuse the samples! }Accordingly, our sample construction is simply as follows. (Assume an ordering among different states of $\calA$, denoted by $\prec$ )

 \begin{align*}
	S(q^\ell) &= \Bigg( \bigcup_{(q_0,0,q) \in \calT} S(q^{\ell-1}_0) \mathbin{\big\backslash}  \bigcup_{\substack{(q_i,0,q) \in \calT \\ q_i \prec q_0}}  \mathcal{L}(q_i^{\ell-1}) \Bigg) \cdot \{0\} \cup  \Bigg( \bigcup_{(q_1,1,q) \in \calT} \mathcal{S}(q_1^{\ell-1}) \mathbin{\big\backslash}
	\bigcup_{\substack{(q_i,1,q) \in \calT \\ q_i \prec q_1}}  \mathcal{L}(q_i^{\ell-1})
	\Bigg) \cdot \{1\}
\end{align*}

Observe that it is possible for $w \in \calL(q_i^{\ell-1})$ and $w \in \calL(q_j^{\ell-1})$ and $(q_i,0,q) \in \calT$ and $(q_j,0,q) \in \calT$. In such a case, having an ordering over states of $\calA$ (say $q_i \prec q_j$) ensures that $w0 \in S(q^{\ell})$ can only happen if $w \in S(q_i^{\ell-1})$, i.e., if it was the case that $w \notin  S(q_i^{\ell-1})$ but $w \in  S(q_j^{\ell-1})$, then $w0$ would not be in $S(q^{\ell})$.
 Essentially, such an update ensures that while there may be many runs for a word $w$ to reach $q^{\ell}$, there is a unique derivation run (formally defined in Section~\ref{section:derivation_run}) for every $w$ such that every prefix of $w$ must be in the corresponding sample sets along that run for $w \in  S(q^{\ell})$.  

To highlight the benefits of the reuse of samples, let us point out that the prior work requires polynomially many calls to union of sets estimation for generation of even one sample -- as we would have to make two union of sets estimation calls for every $k$ (and another polynomial factor to account for accumulation of the errors). On the other hand, reuse of samples requires no calls to union of sets procedures, as we simply extend samples from sets in the previous layer.

Of course, the reuse of samples comes with a price: the need to handle dependence, which requires a dramatically different analysis in contrast to prior work.  This is where we rely on the recently proposed approach of bounding dependence via analysis of derivation paths and coupling with random processes~\cite{MdC24a}.  
Observe that if $(q_0, 0,q_1) \in \calT$ and $(q_0,0,q_2) \in \calT$, then for every $\ell$, the construction of sets $S(q_1^{\ell})$ and $S(q_2^{\ell})$ would be reusing samples from $S(q_0^{\ell-1})$.  

A natural question is to determine how to bound the dependence? To this end, let us begin with random variables of interest. We are interested in the random variable $\mathbbm{1}[w \in S(q^{\ell})]$ for every $w \in \calL(q^{\ell})$. Ideally, we would like all these random variables to be independent -- this is what previous techniques were able to achieve by avoiding the reuse of the samples. In our case, we should perhaps settle for bounding pairwise dependence. To this end, let us focus on the quantity 	$\sum_{w,w' \in \calL(q^{\ell})}\Pr[w \in S(q^{\ell}) \text{ and } w' \in S(q^{\ell})]$ which upper bounds the variance of our estimator. 
Since, for every word $w$, we have a unique derivation run, the dependence between the events $\mathbbm{1}[w \in S(q^{\ell})]$ and $ \mathbbm{1}[w' \in S(q^{\ell})]$ can be captured by the overlap between runs of $w$ and $w'$. In particular, let $\hat{q}$ be the last common prefix state of the runs of $w$ and $w'$, then we can show that 
\begin{align}\label{eq:varintution}
	\Pr[w \in S(q^{\ell}) \text{ and } w' \in S(q^{\ell})] = \left(\frac{1}{N(q^{\ell})}\right)^2 \cdot N(\hat{q})
\end{align}
First, contrast the above expression with what would have been the case if the events $\mathbbm{1}[w \in S(q^{\ell})]$ and $ \mathbbm{1}[w' \in S(q^{\ell})]$ were pairwise independent, in such a case we would have  	$\Pr[w \in S(q^{\ell}) \text{ and } w' \in S(q^{\ell})] = \left(\frac{1}{N(q^{\ell})}\right)^2 $, therefore, we are off by a factor of $N(\hat{q})$, which can grow almost as much as $|\calL(q^{\ell})|$. The key observation is that the number pairs of $w$ and $w'$ whose accepting runs have $\hat{q}$ as last common prefix state can be bounded by $\frac{|\calL(q^{\ell})|^2}{|\calL(\hat{q})|}$. Consequently, we are able to bound the variance of our estimator and ensure $N(q^{\ell})$ is a good approximation of $|\calL(q^{\ell})|$.

Before concluding, we discuss another crucial technical innovation in our work. Since we reuse samples, a significant factor contributing to our time complexity is the membership check during the construction of samples, as we need to verify whether a word $w \in \calL(q^\ell)$. We observe that the structure of our algorithms allows us to amortize the cost of these membership checks through an incremental caching mechanism.

In a standard implementation, membership verification $w \in \calL(q^\ell)$ requires $O(nm^2)$ time, where $m$ is the number of states in the original automaton. When processing the unrolled automaton, each transition $(q',b,q)$ with $q' \in \calQ^\ell$ and $q \in \calQ^{\ell-1}$ potentially triggers $O(|\calQ^{\ell-1}|)$ membership tests. With a standard per-layer caching approach, each sample for $q'$ incurs $O(m)$ queries costing $O(nm^2)$ each, plus $O(m^2)$ constant-time lookups, yielding $O(nm^3)$ cost per sample for membership verification alone. With the FPRAS requiring $O(n^2m\varepsilon^{-2}\log(\delta^{-1}))$ samples, this would result in a prohibitive $O(n^3m^4\varepsilon^{-2}\log(\delta^{-1}))$ runtime. Our key insight is that by constructing samples for $\calQ^\ell$ directly from samples for $\calQ^{\ell-1}$, we can derive the cache for level $\calQ^\ell$ from the cache for level $\calQ^{\ell-1}$ through efficient matrix operations, reducing membership verification to $O(1)$ per sample and eliminating the $nm$ factor. We present an efficient incremental cache construction technique that ensures membership checks dominate the runtime, yielding our improved $O(n^2m^3\varepsilon^{-2}\log(\delta^{-1}))$ complexity. We further explore matrix-based caching strategies that show particular promise for practical implementations on modern hardware architectures.

To provide high-level intuition, we have omitted many technical details. We mention two of them briefly to prepare the reader for rest of the paper:  First, the computation of $N(q)$ differs from prior work because we do not have access to independent samples. As a result, we must work under the constraint of limited independence. To address this, we employ the median-of-means estimator. Secondly, it should be noted that Equation (\ref{eq:varintution}) is not entirely rigorous, as $N(q)$ itself is a random variable. Therefore, similar to~\cite{MdC24a}, our analysis proceeds by coupling it with a random process, as detailed in Section~\ref{sec:randomprocess}.

\paragraph{Organization:} The remainder of this paper is structured as follows. Section~\ref{section:background} introduces notations and preliminaries. Section~\ref{section:derivation_run} presents the notion of derivation run, crucial for both algorithm design and analysis. Section~\ref{section:algorithm} presents the proposed algorithm. Section~\ref{section:analysis} provides a detailed analysis of the algorithm's performance; due to space constraints, some proofs are included in the Appendix.

\section{Notations and Prelminaries}\label{section:background}

Given two integers $m < n$, $[n]$ denotes the set $\{1,2,\dots,n\}$. For $a$, $b$ and $\varepsilon$ three non-negative real numbers, we use $a  \in (1 \pm \varepsilon)b$ to denote $(1 - \varepsilon)b \leq a \leq  (1 + \varepsilon)b$, similarly, $a  \in \frac{b}{1 \pm \varepsilon}$, or $a \in b(1\pm \varepsilon)^{-1}$, stands for $\frac{b}{1 + \varepsilon}\leq a \leq  \frac{b}{1 - \varepsilon}$ (with $\frac{b}{0}$ defined as $\infty$ when $b\neq 0$, and $\frac{0}{0}$ defined as $0$). 

A word $w$ over an alphabet $\Sigma$ is a sequence $(w_1 \dots w_k)$ of length $|w| = k$ where each $w_i$ belongs to $\Sigma$. The empty word $\lambda$ (of length $0$) is the empty sequence. $\Sigma^*$ denotes the set of all words over $\Sigma$. For $w$ and $w'$ in $\Sigma^*$, we denote by $w\cdot w'$ their concatenation. For $S \subseteq \Sigma^*$, we write $S\cdot w = \{w' \cdot w \mid w' \in S\}$.

\paragraph{FPRAS}

For a problem that, given an input $x$ of size $s$, aims at computing a real number $N(x)$, a fully polynomial-time randomized approximation scheme (FPRAS) is an algorithm that, given $x$, $\varepsilon > 0$, and $0 < \delta < 1$, runs in time polynomial in $s$, $1/\varepsilon$ and $\log(1/\delta)$, and returns $\tilde{N}$ with the guarantee that
$
\Pr\left[ \tilde{N} \in (1 \pm \varepsilon)N(x)\right] \geq 1 - \delta.
$

\paragraph{NFA}

In this paper we work with the alphabet $\Sigma = \{0,1\}$. A non-deterministic finite automaton (NFA) is a tuple $\calA = (\calQ,q_I,q_F,\calT)$ where $\calQ$ is a finite set of states, $q_I \in \calQ$ is the initial state, $\calT \subseteq \calQ \times  \Sigma \times \calQ$ is a transition relation, and $q_F \in \calQ$ is the final state. We assume a total order $\prec$ on the states. A \emph{run} of $w$ on $\calA$ is a sequence $\rho = (q_0,\dots,q_{|w|})$ such that $q_0 = q_I$ and, for every every $i < |w|$, $(q_i,w_{i+1},q_{i+1}) \in \calT$; $\rho$ \emph{ends on}, or \emph{reaches}, $q_{|w|}$. A word $w$ is \emph{accepted by} $q \in \calQ$, when there is a run of $w$ on $\calA$ that ends on $q$. The \emph{language} of $\calL(q)$ is the set of words accepted by $q$ and the language of $\calA$ is $\calL(\calA) = \calL(q_F)$. For $n \in \mathbb{N}$, the $n$-th \emph{slice} of $\calA$'s language, $\calL_n(\calA)$, is the set of words of length $n$ in $\calL(\calA)$. For $q \in \calQ$ and $b \in \Sigma$, the $b\textit{-}$\emph{predecessors} of $q$ are given by a sequence $\pred(q, b) = (q' \mid (q', b, q) \in \calT)$ ordered consistently with $\prec$. We also define $\pred(q) = (q' \mid (q', 0, q) \in \calT \text{ or } (q', 1, q) \in \calT)$ again with states ordered according to $\prec$.

\paragraph{Unrolling}

Given an automaton $\calA = (\calQ,q_F,q_I,\calT)$ and $n \in \mathbb{N}$, one constructs the \emph{unrolled automaton} $\calA^{\mathsf{u}}_n$ whose language is $\calL_n(\calA)$ in time $O(n|\calT|)$. One can check whether $\calL_n(\calA)$ is empty in time $O(n|\calT|)$; we assume $\calL_n(\calA) \neq \emptyset$. Since $n$ will be fixed, we will just write $\calA^{\mathsf{u}} = (\calQ^\mathsf{u},q_I^0,q_F^n,\calT^\mathsf{u})$. For each state $q \in \calQ$ and $0\leq \ell \leq n$, if $q$ is reachable in $\calA$ by \emph{some} word of length $\ell$ then $\calQ^\mathsf{u}$ contains a state $q^\ell$ and $\calT^\mathsf{u}$ ensures that $q^\ell$ is reachable by exactly \emph{all} words of length $\ell$ that reach $q$ in $\calA$. Formally, $q^0_I$ is in $\calQ^\mathsf{u}$ and, for every $0 \leq \ell < n$, if $(q_1,b,q_2)$ is in $\calT$ and $q_1^\ell$ is in $\calQ^\mathsf{u}$, then $q_2^{\ell+1}$ is in $\calQ^\mathsf{u}$ and $(q_1^\ell,b,q_2^{\ell+1})$ is in $\calT^\mathsf{u}$. We write $\calQ^\ell = \{q^\ell \mid q \in \calQ \text{ and } q \text{ is reachable by words of length } \ell\,\}$ and $\calQ^{< \ell} = \calQ^0 \cup \dots \cup \calQ^{\ell-1}$ and similarly for $\calQ^{\leq \ell}$, $\calQ^{> \ell}$ and $\calQ^{\geq \ell}$. Note that $\calQ^0$ contains only $q^0_I$. Graphically, we see unrolled automata as directed acyclic graphs (DAGs) where the states are vertices and the transitions are edges labeled by the transition symbol $b$; see for instance Figure~\ref{figure:nfa}. We will again assume some total order $\prec$ on the states of $\calA^\mathsf{u}$. Since in this paper we work only with unrolled automatons, we drop the superscript from the states' name. For $q \in \calQ^\mathsf{u}$, we denote by $\ancestors(q)$ its ancestors set, defined inductively as $\ancestors(q^0_I) = \emptyset$ and $\ancestors(q) = \pred(q) \cup \bigcup_{q' \in \pred(q)} \ancestors(q')$. For instance the ancestors of $q_5$ in Figure~\ref{figure:nfa} are $q_I$, $q_1$ and $q_2$. Note that $q \not\in \ancestors(q)$.
\section{Derivation runs}\label{section:derivation_run} 

A word can have several accepting runs, seen as paths in the DAG corresponding to $\calA^\mathsf{u}$. For instance, in Figure~\ref{figure:nfa}, there are two possible runs, highlighted in red and blue, for the word $(010)$ that end on $q_{10}$. In this section a $b$-transition from $q$ to $q'$ ($b \in \{0,1\}$) is denoted by \textcolor{black}{$q \xrightarrow{b} q'$}. For a run \textcolor{black}{$q_1 \xrightarrow{b_1} q_2 \xrightarrow{b_2} \dots  \xrightarrow{b_{k-1}} q_k$} in $\calA^\mathsf{u}$ and for $1 \leq \ell \leq k$, the $\ell^\text{th}$ prefix of \textcolor{black}{$q_1 \xrightarrow{b_1} q_2 \xrightarrow{b_2} \dots  \xrightarrow{b_{k-1}} q_k$} is $q_1$ if $\ell = 1$ and \textcolor{black}{$q_1 \xrightarrow{b_1} q_2 \xrightarrow{b_2} \dots  \xrightarrow{b_{\ell-1}} q_\ell$} otherwise. Given a run $R$ that ends on $q$ , we denote by $R \xrightarrow{b} q'$ the run that extends $R$ with the transition $q \xrightarrow{b} q'$. For $w$ a word in $\calL(q)$, we map $(w,q)$ to a unique accepting run, called the \emph{derivation run} of $w$ for $q$.

\begin{figure}[b]
\centering
\begin{tikzpicture}[xscale=2.5,yscale=1.4]

\draw[color=blue,opacity=0.3,line width=2mm] (1,-0.33) -- (2,-0.33) -- (3,0);
\draw[color=red,opacity=0.3,line width=2mm] (1,-0.33) -- (2,+0.33) -- (3,0);

\draw[color=red,opacity=0.3,line width=1.2mm] (0,0.04) -- (1,-0.29);
\draw[color=blue,opacity=0.3,line width=1.2mm] (0,-0.04) -- (1,-0.38);

\node[fill=white,draw,circle,minimum size=0.6cm,inner sep=0] (qi) at (0,0) {$q_I$};

\node[fill=white,draw,circle,minimum size=0.6cm,inner sep=0] (q1) at (1,+1.0)  {$q_1$};
\node[fill=white,draw,circle,minimum size=0.6cm,inner sep=0] (q2) at (1,+0.33) {$q_2$};
\node[fill=white,draw,circle,minimum size=0.6cm,inner sep=0] (q3) at (1,-0.33) {$q_3$};
\node[fill=white,draw,circle,minimum size=0.6cm,inner sep=0] (q4) at (1,-1.0)  {$q_4$};

\node[fill=white,draw,circle,minimum size=0.6cm,inner sep=0] (q5) at (2,+1.0)  {$q_5$};
\node[fill=white,draw,circle,minimum size=0.6cm,inner sep=0] (q6) at (2,+0.33) {$q_6$};
\node[fill=white,draw,circle,minimum size=0.6cm,inner sep=0] (q7) at (2,-0.33) {$q_7$};
\node[fill=white,draw,circle,minimum size=0.6cm,inner sep=0] (q8) at (2,-1.0)  {$q_8$};

\node[fill=white,draw,circle,minimum size=0.6cm,inner sep=0] (q9) at (3,+1) {$q_9$};
\node[fill=white,draw,circle,minimum size=0.6cm,inner sep=0] (q10) at (3,0)  {$q_{10}$};
\node[fill=white,draw,circle,minimum size=0.6cm,inner sep=0] (q11) at (3,-1) {$q_{11}$};

\node[fill=white,double,draw,circle,minimum size=0.6cm,inner sep=0] (qf) at (4,0) {$q_F$};

\draw[-latex] (qi) -- (q1);
\draw[-latex] (qi) -- (q2);
\draw[densely dashed,-latex] (qi) -- (q3);
\draw[densely dashed,-latex] (qi) -- (q4);

\draw[-latex] (q1) -- (q5);
\draw[densely dashed,-latex] (q1) -- (q6);
\draw[densely dashed,-latex] (q2) -- (q5);
\draw[-latex] (q2) -- (q7);
\draw[-latex] (q3) -- (q6);
\draw[-latex] (q3) -- (q7);
\draw[densely dashed,-latex] (q3) -- (q8);
\draw[-latex] (q4) -- (q8);

\draw[densely dashed,-latex] (q5) -- (q9);
\draw[-latex] (q5) -- (q10);
\draw[densely dashed,-latex] (q6) -- (q10);
\draw[densely dashed,-latex] (q7) -- (q10);
\draw[-latex] (q7) -- (q9);
\draw[-latex] (q8) -- (q10);
\draw[-latex] (q8) -- (q11);

\draw[-latex] (q9) --(qf);

\draw[-latex] (q10) to [out=20,in=165] (qf);
\draw[densely dashed,-latex] (q10) to [out=-20,in=-165] (qf);
\draw[densely dashed,-latex]  (q11) -- (qf);

\end{tikzpicture}
\caption{An unrolled NFA of length $4$. Solid lines represent $1$-transitions. Dashed lines represent $0$-transitions.}\label{figure:nfa}
\end{figure}
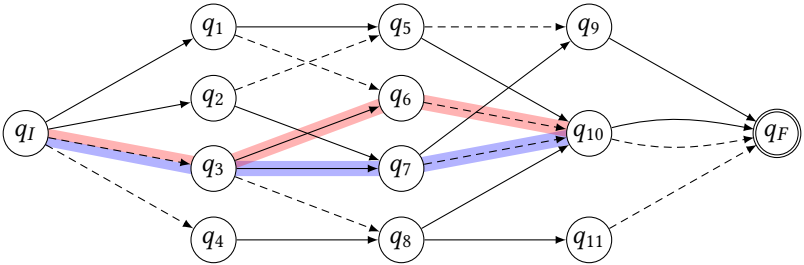

\begin{definition}[Derivation run]
Let $w \in \calL(q)$. The derivation run $\paths(w,q)$ is defined inductively as follows:
\begin{itemize}
\item[•] If $q = q_I$ then $w = \lambda$ and the only derivation run is $\paths(\lambda,q_I) = q_I$.
\item[•] If $q \neq q_I$ then let $b$ be the last symbol of $w$, write $w = w' \cdot b$ and let $q' \in \pred(q,b)$ be the first $b$-predecessor of $q$ such that $w' \in \calL(q')$. Then $\paths(w,q) = \paths(w',q') \xrightarrow{b} q$.
\end{itemize}
\end{definition} 
Going back to Figure~\ref{figure:nfa}, if we assume the state ordering $q_I \prec q_1 \prec q_2 \prec \dots \prec q_{11} \prec q_F$ then $\paths((010),q_{10})$ is the red run $q_I \xrightarrow{0} q_1 \xrightarrow{1} q_6 \xrightarrow{0} q_{10}$.

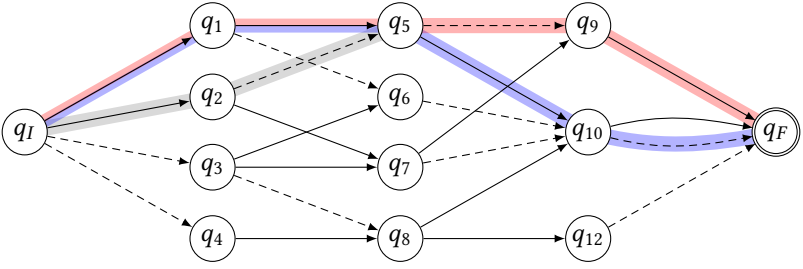
\begin{figure}[h]
\centering
\begin{tikzpicture}[xscale=2.5,yscale=1.4]

\draw[color=blue,opacity=0.3,line width=2mm] (2,1) -- (3,0) to [out=-25,in=-160](4,0);
\draw[color=red,opacity=0.3,line width=2mm] (2,1) -- (3,1) -- (4,0);

\draw[color=gray,opacity=0.3,line width=2mm] (0,0) -- (1,0.33) -- (2,1);

\draw[color=blue,opacity=0.3,line width=1mm] (0,-0.03) -- (1,+0.97) -- (2,+0.97);
\draw[color=red,opacity=0.3,line width=1mm] (0,0.03) -- (1,+1.03) -- (2,+1.03);

\node[fill=white,draw,circle,minimum size=0.6cm,inner sep=0] (qi) at (0,0) {$q_I$};

\node[fill=white,draw,circle,minimum size=0.6cm,inner sep=0]  (q1) at (1,+1.0)  {$q_1$};
\node[fill=white,draw,circle,minimum size=0.6cm,inner sep=0]  (q2) at (1,+0.33) {$q_2$};
\node[fill=white,draw,circle,minimum size=0.6cm,inner sep=0]  (q3) at (1,-0.33) {$q_3$};
\node[fill=white,draw,circle,minimum size=0.6cm,inner sep=0]  (q4) at (1,-1.0)  {$q_4$};

\node[fill=white,draw,circle,minimum size=0.6cm,inner sep=0]  (q5) at (2,+1.0)  {$q_5$};
\node[fill=white,draw,circle,minimum size=0.6cm,inner sep=0]  (q6) at (2,+0.33) {$q_6$};
\node[fill=white,draw,circle,minimum size=0.6cm,inner sep=0]  (q7) at (2,-0.33) {$q_7$};
\node[fill=white,draw,circle,minimum size=0.6cm,inner sep=0]  (q8) at (2,-1.0)  {$q_8$};

\node[fill=white,draw,circle,minimum size=0.6cm,inner sep=0]  (q9) at (3,+1) {$q_9$};
\node[fill=white,draw,circle,minimum size=0.6cm,inner sep=0] (q10) at (3,0)  {$q_{10}$};
\node[fill=white,draw,circle,minimum size=0.6cm,inner sep=0]  (q11) at (3,-1) {$q_{12}$};

\node[fill=white,double,draw,circle,minimum size=0.6cm,inner sep=0] (qf) at (4,0) {$q_F$};

\draw[-latex] (qi) -- (q1);
\draw[-latex] (qi) -- (q2);
\draw[densely dashed,-latex] (qi) -- (q3);
\draw[densely dashed,-latex] (qi) -- (q4);

\draw[-latex] (q1) -- (q5);
\draw[densely dashed,-latex] (q1) -- (q6);
\draw[densely dashed,-latex] (q2) -- (q5);
\draw[-latex] (q2) -- (q7);
\draw[-latex] (q3) -- (q6);
\draw[-latex] (q3) -- (q7);
\draw[densely dashed,-latex] (q3) -- (q8);
\draw[-latex] (q4) -- (q8);

\draw[densely dashed,-latex] (q5) -- (q9);
\draw[-latex] (q5) -- (q10);
\draw[densely dashed,-latex] (q6) -- (q10);
\draw[densely dashed,-latex] (q7) -- (q10);
\draw[-latex] (q7) -- (q9);
\draw[-latex] (q8) -- (q10);
\draw[-latex] (q8) -- (q11);

\draw[-latex] (q9) --(qf);

\draw[-latex] (q10) to [out=20,in=165] (qf);
\draw[densely dashed,-latex] (q10) to [out=-20,in=-165] (qf);
\draw[densely dashed,-latex]  (q11) -- (qf);
\end{tikzpicture}
\caption{The last common prefix state of the blue run and red run is $q_5$.}\label{figure:lcpn}
\end{figure}

In the algorithm presented in Section~\ref{section:algorithm}, sample words in $\calL(q)$ are constructed through their derivation runs. Intuitively, two words are more likely to be sampled together when their derivation runs share a large prefix. Given two runs $R$ and $R'$, we call their \emph{last common prefix state} the deepest state contained in both $R$ and $R'$ such that the two runs are consistent up to that state. 
\begin{definition}[last common prefix state]
Let $R$ and $R'$ be two runs in $\calA^\mathsf{u}$ both starting at $q_I$. The \emph{last common prefix state} of $R$ and $R'$, denoted by $\lcpn(R,R')$, is $k^{\text{th}}$ state of $R$ for the largest possible $k$ such that the $k^{\text{th}}$ prefix of $R$ equals the $k^{\text{th}}$ prefix of $R'$.
\end{definition}

Figure~\ref{figure:lcpn} gives an example where the least common prefix state is $q_5$.
Let $w \in \calL(q)$. We denote by $I(w,q,\ell)$ the set of words $w'$ in $\calL(q)$ whose derivation run $\paths(w',q)$ are shared with $\paths(w,q)$ up to the $\ell^{\text{th}}$ state, and then diverge at the $(\ell+1)^{\text{th}}$ state. In other words, if $q \in \calQ^i$ and $w = (w_1w_2\dots w_i)$ and $\paths(w,q) = q_0 \xrightarrow{w_1} q_1 \xrightarrow{w_2} \dots \xrightarrow{w_i} q_i$ with $q_0 = q_I$ and $q_i = q$ then
\[
I(w,q,\ell) = \{w' \in \calL(q) \mid \lcpn(\paths(w,q),\paths(w',q)) = q_\ell\}. 
\]
The next lemma gives an upper bound on $|I(w,q,\ell)|$ and plays a key role in the analysis, so we provide some intuition using Figure~\ref{figure:lcpn}. Assuming $q_I \prec q_1 \prec q_2 \prec \dots \prec q_{11} \prec q_F$, Figure~\ref{figure:lcpn} shows the derivation run for $w = (1101)$ in red. Two words have the derivation runs diverging from $\paths(w,q_F)$ at $q_5$, namely $w' = (1111)$ and $w'' = (1110)$, see for instance $\paths(w'',q_F)$ in blue. So $I(w,q_F,2) = \{(1111),(1110)\}$. In $\paths(w',q_F)$ and $\paths(w'',q_F)$, replacing the portion of the run up to $q_5$ by any other run reaching $q_5$ creates other words in $\calL(q_F)$. For instance the replacement with the run highlighted in gray in $w'$ and $w''$ yields $(1011)$ and $(1010)$. In fact this construction generates $|\calL(q_5)|\cdot |I(w,q_F,2)|$ distinct words in $\calL(q_F)$, hence $|\calL(q_5)|\cdot |I(w,q_F,2)| \leq |\calL(q_F)|$.

\begin{restatable}{lemma}{derivationPath}\label{lemma:derivation_path}
Let $q \in \calQ^i$, $w \in \calL(q)$, $\paths(w,q) = q_0 \xrightarrow{w_1} \dots \xrightarrow{w_i} q_i$ with $q_i = q$ and $q_0 = q_I$ and let $0 \leq \ell \leq i$, then $|I(w,q,\ell)| \leq \frac{|\calL(q)|}{|\calL(q_\ell)|}$.
\end{restatable}
\begin{proof}
Let $I(w,q,\ell) = \{w^1,w^2,\dots\}$. Let $\omega = (w_1,\dots,w_\ell)$ be the restriction of $w$ to the first $\ell$ symbols. By definition, every $w^i$ is of the form $\omega \cdot \omega^i$ for some $\omega^i \in \{0,1\}^{|w|-i}$. If $\omega'$ is a word in $\calL(q^\ell_w)$, then every $\omega' \cdot \omega^i$ is in the language $\calL(q)$. Since the words $w^i$ differ on the last $|w| - \ell$ symbols, the $\omega^i$'s are pairwise distinct, and therefore $\{\omega' \cdot \omega^i\}_i$ is a set of $|I(w,q,\ell)|$ distinct words in $\calL(q)$. Considering all $|\calL(q^\ell_w)|$ possible choices of words for $\omega'$, we find that $\{\omega' \cdot \omega^i \mid \omega' \in \calL(q^\ell_w),\,i \leq |I(w,q,\ell)|\}$ is a set of $|\calL(q^\ell_w)|\cdot|I(w,q,\ell)|$ distinct words in $\calL(q)$.
\end{proof}

\section{Algorithm}\label{section:algorithm}

Our goal is to compute a good approximation of $|\calL_n(\calA)|$. An $O(n|\calQ|^2)$-time preprocessing is enough to determine whether there exist words of length $n$ accepted by $\calA$, thus we assume that $\calL_n(\calA)$ is not empty. Our FPRAS algorithm computes the unrolled automaton $\calA^{\mathsf{u}} = (\calQ^{\mathsf{u}},q_I,q_F,\calT^{\mathsf{u}})$ for that $n$ and processes its states one by one, with the predecessors of $q$ processed before $q$. For each state $q \in \calQ^\mathsf{u}$ the algorithm computes $p(q)$, which seeks to estimate $|\calL(q)|^{-1}$, and polynomially-many sets $S^1(q),\dots,S^\gamma(q)$, all subsets of $\calL(q)$. We call them sample sets. To make the connection to earlier work easier we also define the variable $N(q)$ as $p(q)^{-1}$, though it is not used in the analysis. At the high level, for every state $q$ with $\pred(q) = (q_1,\dots,q_k)$, the algorithm does the following:
\begin{enumerate}[leftmargin=*]
\item[•] use the samples sets $(S^r(q_i))_{\textcolor{black}{i \in [k],r \in [\gamma]}}$ and the $(p(q_i))_{\textcolor{black}{i \in [k]}}$ to compute the estimate $p(q)$
\item[•] construct the sample sets $(S^r(q))_{\textcolor{black}{r \in [\gamma]}}$ using $p(q)$ and the predecessors' sets $(S^r(q_i))_{\textcolor{black}{i \in [k],r \in [\gamma]}}$ 
\end{enumerate}
The algorithm stops after processing $q_F$ and returns $N(q_F)$. The sample sets $(S^r(q))_{r \in [\gamma]}$ and the values $p(q)$ are determined using randomness and thus are random variables. A key difference between our FPRAS and the FPRAS schemes of~\cite{ACJR19,MeelCM24} is that the samples for $q$ are built directly from the samples for $q_i$. This renders the variables $(S^r(q))_{r \in [\gamma]}$ and $p(q)$ quite dependent of one another. Our FPRAS works towards ensuring that ``$\Pr[w \in S^r(q)] = p(q)$''\footnote{We will explain later that this expression is not quite correct.} holds for every $q \in \calQ^\mathsf{u}$ and $w \in \calL(q)$. Thus, if $p(q)$ is a good estimate of $|\calL(q)|^{-1}$, then $S^r(q)$ should be small in expectation.

\subsection{General Description}

We describe how on an estimator for $|\calL(q)|$ is obtained for every state $q$ of the unrolled automaton. For starter, consider the problem of approximating $|\calL(q)|$ where $\pred(q) = (q_1,\dots,q_k)$ and when a sample set $S(q_i) \subseteq \calL(q_i)$ is available for every $i$ with the guarantee that $\Pr[w \in S(q_i)] = \rho$ holds for every $w \in \calL(q_i)$. We compute $\hat S(q)  = \union(q,S(q_1),\dots,S(q_k))$ as follows: for every $1 \leq i \leq k$, $b \in \{0,1\}$ and $w \in S(q_i)$, the word $w \cdot b$ is added to $\hat S(q)$ if and only if $q_i$ is the first $b$-predecessor of $q$ (with respect to $\prec$) accepting $w$. Simple computations show that $\Pr[w \in \hat S(q)] = \rho$ holds for every $w \in \calL(q)$, and therefore $\rho^{-1}|\hat S(q)|$ is an unbiased estimate of $|\calL(q)|$ (i.e., its expected value equals $|\calL(q)|$). Efficient ways of doing the $\union$ are presented in the next section.

\begin{algorithm}
$S' = \emptyset$\\
\For{$b \in \{0,1\}$}{
	let $J$ be the subset of $\{1,\dots,k\}$ such that $(q_j \mid j \in J) = b\text{-}\pred(q)$	\\
	\For{$j \in J$ \textup{and} $w \in S_j$}{
		\textbf{if} $w \not\in \calL(q_\ell)$ for every $\ell < j$ with $\ell \in J$ \textbf{then} 
			add $w \cdot b$ to $S'$
	}
}
\Return{$S'$}
\caption{$\union(q,S_1,\dots,S_k)$ \hfill (informal)
\\
where $q$ has $k$ predecessors ordered as $q_1 \prec \dots \prec q_k$ and $S_i \subseteq \calL(q_i)$ for all $1 \leq i \leq k$}\label{algorithm:union}
\end{algorithm}

Now suppose that different sampling probabilities for every the sets $S(q_i)$, that is, for every $i$ we only know some number $p(q_i)$ such that $\Pr[w \in S(q_i)] = p(q_i)$ holds for all $w \in \calL(q_i)$. Let $\rho(q) = \min(p(q_1),\dots,p(q_k))$. Then we normalize the probabilities using the $\reduce$ procedure. That is, we compute $\bar S(q,q_i) = \reduce(S(q_i), \rho(q)/p(q_i))$, a copy of $S(q_i)$ where each word is kept with probability $\rho(q)/p(q_i)$. We have that $\Pr[w \in \bar S(q,q_i)] = \rho(q)$ for every $w \in \calL(q_i)$ so, with $\hat S(q) = \union(q,\bar S(q,q_1),\dots,\bar S(q,q_k))$, we have that $\rho(q)^{-1}|\hat S(q)|$ is an unbiased estimate of $|\calL(q)|$.

\begin{algorithm}
$S' \leftarrow \emptyset$\\
\For{$s \in S$}{
	add $s$ to $S'$ with probability $p$\\
}
\Return{$S'$}
\caption{$\reduce(S,p)$ with $p \in [0,1]$}
\end{algorithm}

To increase the probability that the estimate is a close to $|\calL(q)|$, we use the ``median of means'' technique. Given two integers $\ns$ and $\nt$ and $\gamma = \ns\nt$, instead of one sample set $S(q_i)$ we now have several sample sets $S^1(q_i),S^2(q_i),\dots,S^{\gamma}(q_i)$ all verifying $\Pr[w \in S^r(q_i)] = p(q_i)$. We define $\bar S^r(q,q_i)$ similarly to $\bar S(q,q_i)$ and $\hat S^r(q)$ similarly to $\hat S(q)$. Each $\rho(q)^{-1}|\hat S^r(q)|$ is an estimate of $|\calL(q)|$. We partition the $\ns\nt$ estimates into $\nt$ batches of $\ns$ estimates, compute the average estimate over each batch, and take the median of the average estimates. For judicious values of $\ns$ and $\nt$, this gives an estimate more tightly concentrated around $|\calL(q)|$ (though not unbiased).

We now describe the core of our FPRAS: \approxMCNFAcore. For $q = q_I$, we just set $p(q_I)$ to the correct value $1$ (Line~c\ref{line:p_q_init}) and all sample sets $S^r(q_I)$ to $\{\lambda\}$ (Line~c\ref{line:S_q_init}). Then for each state $q$ with predecessors $q_1,\dots,q_k$, we call $\estimateAndSample(q)$. At that point, all $p(q_i)$'s and all the $S^r(q_i)$'s for $r \in [\gamma]$ and $i \in [k]$ have been computed and $\estimateAndSample(q)$ does the estimate computation described above: it reduces the $S^r(q_i)$ to $\bar S^r(q,q_i)$ (Line~e\ref{line:normalize}), then computes the $\hat S^r(q)$ (Line~e\ref{line:union}) and compute the median of means estimate (Lines~e\ref{line:mm_start}--e\ref{line:mm_end}). The inverse of the median of means estimate is stored in $\hat \rho(q)$. Note that $|\calL(q)| \geq |\calL(q_i)|$ holds for every $i \in [k]$ so, if each $p(q_i)$ is a good estimate $|\calL(q_i)|^{-1}$, then it makes sense to force the estimate of $|\calL(q)|^{-1}$ to be smaller than $\rho(q) = \min(p(q_1),\dots,p(q_k))$. So we take $p(q) = \min(\rho(q),\hat \rho(q))$ (Line~e\ref{line:take_min}). 

$\estimateAndSample(q)$ finishes by computing the sets $S^r(q)$ from the sets $\hat S^r(q)$. Since every word $w \in \calL(q)$ is in $\hat S^r(q)$ with probability $\rho(q)$, and since $p(q) \leq \rho(q)$, we compute $S^r(q)$ $\reduce(\hat{S}^r(q),p(q)/\rho(q))$. Hence the sampling probability $\Pr[w \in S^r(q)] = p(q)$.

\SetNlSty{textbf}{e}{}
\begin{algorithm}
\tikzmark{A}
$\rho(q) = \min(p(q_1),\dots,p(q_k))$\\
\For{$1 \leq r \leq \ns \nt$ \textup{and} $1 \leq i \leq k$}{
		$\bar S^r(q,q_i) = \reduce\big(S^r(q_i),\frac{\rho(q)}{p(q_i)}\big)$\label{line:normalize}
}
\tikzmark{B}
\For{$1 \leq r \leq \ns \nt$}{
	$\hat{S}^r(q) = \union\left(q,\bar S^r(q,q_1),\dots,
	\bar S^r(q,q_k)\right)$\label{line:union}
}
\tikzmark{C}
\For{$1 \leq j \leq \nt$}{\label{line:mm_start}
	$M^j(q) = \frac{1}{\ns\rho(q)}\sum_{r = 1}^{\ns}|\hat{S}^{\ns(j-1) + r}(q)|$
}
$\hat\rho(q) = \median(M^1(q),\dots,M^{\nt}(q))^{-1}$\label{line:mm_end} \\
$p(q) = \min(\rho(q),\hat\rho(q))$\label{line:take_min}\\
$N(q) = \frac{1}{p(q)}$\\
\tikzmark{D}
\For{$1 \leq r \leq \ns\nt$}{
 	$S^r(q) = \reduce\big(\hat{S}^r(q),\frac{p(q)}{\rho(q)}\big)$
\tikzmark{E}\begin{tikzpicture}[remember picture, overlay]
\def\x{6};
\def\o{0.45};
\draw[very thick,decorate,decoration={calligraphic brace,amplitude=6pt,raise=4pt}] ($(pic cs:A) + (\x+0.5, \o-0.15)$)   -- ($(pic cs:B) + (\x+0.5,\o)$);
\node[align=left, anchor=west,font=\footnotesize] at ($(pic cs:A) + (\x+1,-0.4)$)   {normalize sampling probabilities \\ (probability that a word is in $\bar S^r(q,q_i)$ equals $\rho(q)$)};
\draw[very thick,decorate,decoration={calligraphic brace,amplitude=6pt,raise=4pt}] ($(pic cs:B) + (\x+0.5, \o-0.15)$) -- ($(pic cs:C) + (\x+0.5,\o)$);
\node[align=left, anchor=west,font=\footnotesize] at ($(pic cs:B) + (\x+1,-0.1)$)   {generate samples for $q$, keeping only those \\ obtained via their derivation run  \\ (probability that a word is in $\hat S^r(q)$ equals $\rho(q)$)};
\draw[very thick,decorate,decoration={calligraphic brace,amplitude=6pt,raise=4pt}] ($(pic cs:C) + (\x+0.5, \o-0.15)$)   -- ($(pic cs:D) + (\x+0.5,\o)$);
\node[align=left, anchor=west,font=\footnotesize] at ($(pic cs:C) + (\x+1,-0.8)$) {compute the estimate for $|\calL(q)|^{-1}$ with a \\ median of means};
\draw[very thick,decorate,decoration={calligraphic brace,amplitude=6pt,raise=4pt}] ($(pic cs:D) + (\x+0.5, \o-0.15)$)   -- ($(pic cs:E) + (\x+0.5-4.55,\o-0.75)$);
\node[align=left, anchor=west,font=\footnotesize] at ($(pic cs:D) + (\x+1,-0.25)$) {reduce the sample sets for $q$ to ensure that the\\
sampling probability is $p(q)$\\ (probability that a word is in $S^r(q)$ equals $p(q))$};
\end{tikzpicture}}
\caption{$\estimateAndSample(q)$ with $\pred(q) = (q_1,\dots,q_k)$}
\end{algorithm}

\SetNlSty{textbf}{c}{}
\begin{algorithm}
$p(q_I) = 1$ \label{line:p_q_init}\\ 
$\computeCache{0}$\\
\For{$1 \leq r \leq \ns\nt$}{
	$S^r(q_I) = \{\lambda\}$\label{line:S_q_init}\\
}
\For{$1 \leq i \leq n$}{
	$\computeCache{i}$\\
	\For{$q \in \calQ^i$}{
		$\estimateAndSample(q)$\\ 
		\textbf{if }$\sum_{r \in [\ns\nt], q \in \calQ^\mathsf{u}} |S^r(q)| \geq \theta$ \textbf{then return } $0$\label{line:interrupt}
	}
	$\updateCache{i}$\\
}
\Return{$N(q_F)$}
\caption{$\approxMCNFAcore(\calA^{\mathsf{u}},n,\ns,\nt,\theta)$}
\end{algorithm}

To ensure a polynomial running time, $\approxMCNFAcore$ terminates and returns $0$ as soon as the number of samples grows too large (Line~c\ref{line:interrupt}). We will show that the probability of terminating this way is small and that, for well-chosen parameter values, $\approxMCNFAcore$ returns a good estimate of $|\calL(\calA^\mathsf{u})|$ with probability at least $3/4$. The full FPRAS $\approxMCNFA$ amplifies this probability to $1 - \delta$ by returning the median output of several independent runs of $\approxMCNFAcore$.

\SetNlSty{textbf}{}{}
\begin{algorithm}
compute the unrolled NFA $\calA^{\mathsf{u}} = (\calQ^{\mathsf{u}},q_I,q_F,\calT^{\mathsf{u}})$ of $\calA$ for $n$ \\
$\kappa = \frac{\varepsilon}{1 + \varepsilon}$, $\ns = \lceil 4(n+1)\frac{(1+\kappa)^2}{\kappa^2(1-\kappa)} \rceil = \lceil 4(n+1)\frac{(1+2\varepsilon)^2(1+\varepsilon)}{\varepsilon^2} \rceil$, $\nt = \lceil 8\ln(16n|\calQ|) \rceil$, $n_u = \lceil 8\ln(1/\delta)\rceil$, $\theta = 16\ns\nt n(1+\kappa)|\calQ|$\\
\For{$1 \leq j \leq n_u$}{
	$\mathsf{est}_j = \approxMCNFAcore(\calA^{\mathsf{u}},n,\ns,\nt,\theta)$\\
}
\Return $\median(\mathsf{est}_1,\dots,\mathsf{est}_{n_u})$ 
\caption{$\approxMCNFA(\calA = (\calQ,q_I,q_F,\calT),n,\varepsilon,\delta)$}
\end{algorithm}

\subsection{Caching}

Algorithm~\ref{algorithm:union} shows what is expected of $\union$ but does not suggest any implementation. Given a word $w$ and a state $q$ in the unrolled automaton $\calA^\mathsf{u} = (\calQ^\mathsf{u},q_I,q_F,\calT^\mathsf{u})$, it takes time linear in $|\calT^\mathsf{u}|$ to determine whether $w \in \calL(q)$. From the DAG perspective, this is a reachability problem. Thus, there is an implementation of $\union(q,S_1,\dots,S_k)$ in time linear in $k\cdot |\calT^{\mathsf{u}}| \cdot \sum_{i} |S_i|$. In this section, we explain how to achieve better performances using caching. For every $0 \leq i \leq n$, the cache takes the form of two matrices $\cache_i$ and $\cache'_i$. Before going through all states of $\calQ^i$, the matrix $\cache'_i$ is computed by $\computeCache{i}$ using $\cache_{i-1}$. Once all states of $\calQ^i$ have been processed, $\cache'_i$ is updated by $\updateCache{i}$ to give $\cache_i$.

We denote by $\calS^i = \bigcup_{r \in [\ns\nt]} \bigcup_{q \in \calQ^i} S^r(q)$ the set of all words sampled for states in $\calQ^i$. We present two caching schemes. In both schemes, $\cache_{i}$ is a $|\calS^i| \times |\calQ^i|$ matrix. Rows are identified by sampled words and columns are identified by states, so we write $\cache_i(w,q)$ with $w \in \calS^i$ and $q \in \calQ^i$ for the entry at row $w$ and column $q$. We say that $\cache_i$ is \emph{correct} when, for all $w \in \calS^i$, $\cache_i(w,q)$ equals $1$ if $w \in \calL(q)$ and $0$ otherwise, formally:
\begin{align*}\label{equation:cache_correctness}
\cache_i(w,q) = \mathbbm{1}(w \in \calL(q))
\end{align*}

\subsubsection{First caching scheme}

Assuming $\cache_{i-1}$ is correct, $\union(q,S_1,\dots,S_k)$ can be implemented in time linear in $k\cdot \sum_{i} |S_i|$ with Algorithm~\ref{algorithm:union_1}.
\begin{algorithm}
$S' = \emptyset$\\
\For{$b \in \{0,1\}$}{
	let $J$ be the subset of $\{1,\dots,k\}$ such that $(q_j \mid j \in J) = b\text{-}\pred(q)$	\\
	\For{$j \in J$ \textup{and} $w \in S_j$}{
		\textbf{if} $\cache_{i-1}(w,q_\ell) = 0$ for all $\ell \in J$ with $\ell < j$ \textbf{then} add $w \cdot b$ to $S'$\\
	}
}
\Return{$S'$}
\caption{$\union(q,S_1,\dots,S_k)$ with $q \in \calQ^i$ \hfill (1st version)}\label{algorithm:union_1}
\end{algorithm}

Now, how to compute a correct cache? For $i = 0$, $\cache_0$ is the unit $1 \times 1$ matrix: $(1)$. It is correct since $\cache_0(\lambda,q_I) = 1$. For $0 < i \leq n$, the intermediate matrix $\cache'_i$ is constructed from $\cache_{i-1}$ with a matrix product. For $b \in \{0,1\}$, let $\transition^b_{i-1,i}$ be the $|\calQ^{i-1}| \times |\calQ^i|$ $b\textit{-}$transition matrix from $\calQ^{i-1}$ and $\calQ^i$. Formally, for every $(q,q') \in \calQ^{i-1} \times \calQ^i$, $\transition^b_{i-1,i}(q,q') = \mathbbm{1}(q' \in b\text{-}\pred(q))$. Then $\cache'_i$ is computed as
\begin{equation}\label{equation:cache_matrix_product}
\cache'_i = \begin{pmatrix}
\cache_{i-1} \times \transition^0_{i-1,i} \smallskip
\\
\cache_{i-1} \times \transition^1_{i-1,i} 
\end{pmatrix}. \tag{$\dagger$}
\end{equation}
This is a $(2\cdot|\calS^{i-1}|)\times|\calQ^i|$ matrix. The first $|\calS^{i-1}|$ rows correspond to the words $w \cdot 0$ for $w \in \calS^{i-1}$ and the last $|\calS^{i-1}|$ rows to the words $w \cdot 1$. We say that $\cache'_i$ is \emph{correct} when $\cache'_i(w,q) = \mathbbm{1}(w \in \calL(q))$ holds for all $w \in (\calS^{i-1} \cdot 0) \cup (\calS^{i-1} \cdot 1)$ and $q \in \calQ^i$. It is not hard to see that if $\cache_{i-1}$ is correct, then so is $\cache'_i$. Since $\calS^i \subseteq (\calS^{i-1} \cdot 0) \cup (\calS^{i-1} \cdot 1)$, constructing $\cache_i$ amounts to keeping those rows of $\cache'_i$ for words in $\calS^i$. This is done in $\updateCache{i}$. It is clear that the correctness of $\cache'_i$ carries to $\cache_i$.

\subsubsection{Second caching scheme}
\begin{algorithm}
	$S' = \emptyset$\\
	\For{$b \in \{0,1\}$}{	
		let $J$ be the subset of $\{1,\dots,k\}$ such that $(q_j \mid j \in J) = b\text{-}\pred(q)$	\\
		\For{$j \in J$ \textup{and} $w \in S_j$}{
			\textbf{if} $\cache'_i(w \cdot b,q) \in [2^{k-j},2^{k-j+1})$ \textbf{then} add $w\cdot b$ to $S'$
		}
	}
	\Return{$S'$}
	\caption{$\union(q,S_1,\dots,S_k)$ with $q \in \calQ^i$ \hfill (2nd version)}\label{algorithm:union_2}
\end{algorithm}
We now describe a second caching strategy that further simplifies the implementation of $\union$. 
Here $\cache'_i$ is still a $(2\cdot|\calS^{i-1}|)\times|\calQ^i|$ matrix but its entries are not limited to $0$ and $1$. $\cache_0$ is again defined as the unit $1 \times 1$ matrix; we assume again that $\cache_{i-1}$ is correct and we again compute $\cache'_i$ using~(\ref{equation:cache_matrix_product}) but with another matrix $\transition^b_{i-1,i}$. Here we define $\transition^b_{i-1,i}$ as follows: for $q \in \calQ^{i-1}$ and $q' \in \calQ^i$ with $b\text{-}\pred(q') = (q_1,\dots,q_k)$, $\transition^b_{i-1,i}(q,q') = 0$ when $q$ is not in $b\text{-}\pred(q')$, and $\transition^b_{i-1,i}(q,q') = 2^{k - j}$ when $q = q_j$.

Consider a word $w \in \calS^{i-1}$. If $\cache_{i-1}$ is correct then $\cache'_i(w \cdot b, q) = 0$ when $w \cdot b \not\in \calL(q)$. But now, we also have that if $w \cdot b \in \calL(q)$ and, if $q_j$ is the first $b\textit{-}$ predecessor of $q$ with $w \in \calL(q_j)$, then $\cache'_i(w \cdot b,q)$ equals $2^{k-j}$ plus a sum of distinct powers of two whose exponents are strictly smaller than $k - j$. So $\cache'_i(w \cdot b,q)$ is between $2^{k-j}$ and $2^{k-j+1} - 1$. This means that we just need to read $\cache'_i(w \cdot b,q)$ and determine the interval $[2^a,2^{a+1})$ in which it lies to determine the index of the first $b\textit{-}$predecessor of $q$ whose language contains $w$. Hence the implementation of $\union$ in Algorithm~\ref{algorithm:union_2}.

Next, to compute $\cache_i$ in $\updateCache{i}$, we replace in $\cache'_i$ every non-zero entry by $1$ and remove all rows for words not in $\calS^i$. Since $\cache'_i(w,q) = 0$ if and only if $w \in \calL(q)$, we have that $\cache_i$ is correct. Compared to the first caching scheme, rather than looking in $\cache_{i-1}$ to find the earliest predecessor, we encode that information in $\cache_i$. This way, we replace loops over the predecessors of $q$ with a interval check. In theory, the interval check hides loops and the worst-case running time remains linear in $k\cdot \sum_{i}|S_i|$, but there may be practical benefits in using this version of $\union$.

\section{Analysis}\label{section:analysis}

Our algorithm is designed so that, for every $w \in \calL(q)$, we have something like ``$\Pr[w \in S^r(q)] = p(q)$''. But this equality makes no sense because $\Pr[w \in S^r(q)]$ is a fixed real value whereas $p(q)$ is a random variable. In addition we manipulate the sets $S^1(q),S^2(q),\dots$ in the median of means estimator as if they were independent, which is not exactly true. We circumvent these issues by using a random process with new variables that behave more nicely than $S^r(q)$ and $p(q)$. The random process simulates the variant of $\approxMCNFAcore$ where there is no bound on the size of the sets $S^r(q)$ (so without Line~\ref{line:interrupt}). We call it $\approxMCNFAcore^*$. The process play all possible runs of $\approxMCNFAcore^*$ at once and simulates $S^r(q)$ by a different variable for each possible run. A coupling argument then allows us to replace $S^r(q)$ by one of these variables assuming enough knowledge on the run of the algorithm up to $q$. This knowledge is encoded in what we call a \emph{history} for $q$.

\subsection{History}

A \emph{history} $h$ for a set of states $Q \subseteq \calQ^\mathsf{u}$ is a mapping
$
h : Q \rightarrow \mathbb{Q}$.
$h$ is \emph{realizable} when there exists a run of $\approxMCNFAcore^*$ that gives the value $h(q)$ to $p(q)$ for every $q \in Q$. Such a run is said \emph{compatible} with $h$. Two histories $h$ for $Q$ and $h'$ for $Q'$ are \emph{compatible} when $h(q) = h'(q)$ for all $q \in Q \cap Q'$. Compatible histories can be merged into an history $h \cup h'$ for $Q \cup Q'$. For $q \in Q$ and $t \in \mathbb{Q}$, we write $h \cup (q \mapsto t)$ to refer to the history $h$ augmented with $h(q) = t$. We only study histories realizable for sets $Q$ that are closed under ancestry, that is, if $q \in Q$ and $q' \in \ancestors(q)$, then $q' \in Q$. Thus we abuse terminology and simply refer to a history for $\ancestors(q)$ as a history for $q$. The only history for $q_I$ is the vaccuous history $h_\emptyset$ for $Q = \emptyset$ (because no ancestors). For $q \in \calQ^\mathsf{u}$, the random variable $H(q)$ is the history for $q$ obtained when running $\approxMCNFAcore^*$.

\subsection{Random Process}\label{sec:randomprocess} 

The random process comprises $n_sn_t$ independent copies identified by the superscript $r$. For $q \in \calQ^\mathsf{u}$, $t \in \mathbb{Q}$ and $h$ a realizable history for $q$, we have several random variables $\vY{r}{h}{t}{q}$ with domain all possible subsets of $\calL(q)$.  $\vY{r}{h}{t}{q}$ simulates $S^r(q)$ in the situation where the value $p(q')$ for each ancestor $q'$ of $q$ has been set to $h(q')$ and where $p(q)$ is set to $t$ (so $t$ is restricted to values that can be given to $p(q)$ by the algorithm under $h$). The variables $\vY{r}{h}{t}{q}$ are defined inductively.
\begin{enumerate}[leftmargin=*]
\item[•] If $q = q_I$ then $\calL(q) = \{\lambda\}$ ($\lambda$ the empty word) and only $\vY{r}{h_\emptyset}{1}{q} = \{\lambda\}$ is defined.
\item[•] If $q \neq q_I$ then let $\pred(q) = (q_1,\dots,q_k)$. For every $\vY{r}{h_1}{t_1}{q_1},\dots,\vY{r}{h_k}{t_k}{q_k}$ such that the histories $h_1,\dots,h_k$ are pairwise compatible, let $h = h_1 \cup \dots \cup h_k \cup (q_1 \mapsto t_1,\dots,q_k \mapsto t_k)$, $t_{\min} = \min(t_1,\dots,t_k)$, and
$
\vZ{r}{h}{q} = \union(q,\reduce(\vY{r}{h_1}{t_1}{q_1},\frac{t_{\min}}{t_1}),
\dots,
\reduce(\vY{r}{h_k}{t_k}{q_k}),\frac{t_{\min}}{t_k})).
$
Now, for all $t \leq t_{\min}$, define
$
\vY{r}{h}{t}{q} = \reduce(\vZ{r}{h}{q},\frac{t}{t_{\min}}).
$
\end{enumerate}
\noindent The variables $\vZ{r}{h}{q}$ and $\vY{r}{h}{t}{q}$ simulate $\hat{S}^r(q)$ and $S^r(q)$ when the algorithm run is compatible with history $h$ and when the value $t$ is computed for $p(q)$. This is expressed formally in Lemma~\ref{lemma:not_so_black_magic}.

\begin{restatable}{lemma}{notSoBlackMagic}\label{lemma:not_so_black_magic}
For every $q \in \calQ^\mathsf{u}$ and $R \subseteq [n_sn_t]$, we have $(H(q),(\hat{S}^r(q))_{r \in R}) = (H(q),(\vZ{r}{H(q)}{q})_{r \in R)})$ and $(H(q),p(q),(S^r(q))_{r \in R}) = (H(q),p(q),(\vY{r}{H(q)}{p(q)}{q})_{r \in R})$.
\end{restatable}
\noindent  Lemma~\ref{lemma:not_so_black_magic} is proved in appendix. The equalities between distribution say that for every $(A^r \subseteq \calL(q) \mid r \in R)$ we have:
\begin{align*}
&\Pr\left[H(q) = h, \,p(q) = t ,\, \bigwedge\nolimits_{r \in R}S^{r}(q) = A^r\right]
= 
\Pr\left[H(q) = h,\,p(q) = t,\, \bigwedge\nolimits_{r \in R}\vY{r}{h}{t}{q} = A^r\right]
\\
&\Pr\left[H(q) = h \text{ and } \bigwedge\nolimits_{r \in R} \hat{S}^{r}(q) = A^r\right]
= 
\Pr\left[H(q) = h \text{ and } \bigwedge\nolimits_{r \in R} \vZ{r}{h}{q} = A^r\right].
\end{align*}

All variables defined within the random process are built independently from the algorithm, in particular they are independent of the variables $H(q)$, $S^r(q)$ and $\hat S^r(q)$. 

\begin{fact}\label{fact:independence}
Events over $\{\vY{r}{h}{t}{q},\vZ{r}{h}{q}\}_{r,t,h,q}$ are independent of events over $\{H(q),S^r(q),\hat{S}^r(q)\}_{r,q}$.
\end{fact}

Whereas the size of $S^r(q)$ and $S^{r'}(q')$ may not be independent when $r \neq r'$ (even when $q = q'$), the $n_sn_t$ copies of the random process do not interact with one another and therefore variables defined within the random process for different $r$ are independent.

\begin{fact}\label{fact:more_independence}
Let $I \subseteq [n_sn_t]$ and consider a variable $\vY{i}{h_i}{t_i}{q_i}$ (resp. $\vZ{i}{h_i}{q_i}$) for every $i \in I$. The variables $\{\vY{i}{h_i}{t_i}{q_i}\}_{i \in I}$ (resp. $\{\vZ{i}{h_i}{q_i}\}_{i \in I}$) are mutually independent.  
\end{fact}

In the random process we have a proper version of the fallacious equality ``$\Pr[w \in S^r(q)] = p(q)$'', as stated in the following two lemmas (proved in appendix).

\begin{restatable}{lemma}{probaFirstOrder}\label{lemma:proba_first_order}
For every $\vY{r}{h}{t}{q}$ and $w \in \calL(q)$, we have
$\Pr[w \in \vY{r}{h}{t}{q}] = t$.
In addition, if $\pred(q) = (q_1,\dots,q_k)$ then $\Pr[w \in \vZ{r}{h}{q}] = \min(h(q_1),\dots,h(q_k))$. 
\end{restatable}

\begin{restatable}{lemma}{probaSecondOrder}\label{lemma:proba_second_order}
For every $\vY{r}{h}{t}{q}$ and $w,w' \in \calL(q)$, $w \neq w'$, let $t^* = h(\lcpn(\paths(w,q),\paths(w',q))$. We have
$
\Pr[w \in \vY{r}{h}{t}{q},\, w' \in \vY{r}{h}{t}{q}] \leq t^2/t^*
$ 
and, with $\rho_h(q) = \min_{q' \in \pred(q)}(h(q'))$ we have
$
\Pr[w \in \vZ{r}{h}{q},\,w' \in \vZ{r}{h}{q}] \leq \rho_h(q)^2/t^*.
$
\end{restatable}

\subsection{Analysis of the FPRAS}
We now conduct the analysis of {\approxMCNFA}. The hardest part to analyze is the core algorithm {\approxMCNFAcore} and to this we will first focus on the analysis of  $\approxMCNFAcore^*$, which, as discussed above, is $\approxMCNFAcore$ without the terminating condition of Line~c\ref{line:interrupt}, so the same algorithm but where $\sum_{r,q} |S^r(q)|$ can grow big.

Recall that $\calA^\mathsf{u} = (\calQ^\mathsf{u},q_I,q_F,\calT^\mathsf{u})$ is an unrolled NFA for words in $\{0,1\}^n$, $\varepsilon > 0$ and $\kappa$, $n_s$, $n_t$ and $\theta$ are defined as in $\approxMCNFA$. The following lemmas provide bound on the correctness of $p(q)$ as well as the size of $S^r(q)$ for $\approxMCNFAcore^*$.  We defer their proofs to later subsections. 

\begin{lemma}\label{lemma:proba_p(q)}
	The probability that $\approxMCNFAcore^*(\calA^\mathsf{u},n,n_s,n_t,\theta)$ computes $p(q) \not\in (1 \pm \kappa)|\calL(q)|^{-1}$ for some $q \in \calQ^\mathsf{u}$ is at most $1/16$.
\end{lemma}

\begin{restatable}{lemma}{probaSq}\label{lemma:proba_S(q)}
	The probability that $\approxMCNFAcore^*(\calA^\mathsf{u},n,n_s,n_t,\theta)$ constructs sets $S^r(q)$ such that $\sum_{r \in [n_sn_t],q \in \calQ^\mathsf{u}} |S^r(q)| \geq \theta$ is at most $1/8$.
\end{restatable}

Combining the above lemmas, we can now state the correctness and runtime of {\approxMCNFA}.

\begin{restatable}{lemma}{mainResultCore}\label{lemma:main_result_core}
	Let $\calA^\mathsf{u} = (\calQ^\mathsf{u},q_I,q_F,\calT^\mathsf{u})$ be an unrolled NFA recognizing words in $\{0,1\}^n$. Let $m = \max_i |\calQ^i|$, $\varepsilon > 0$, and $\kappa$, $n_s$, $n_t$ and $\theta$ be defined as in $\approxMCNFA$, then $\approxMCNFAcore(\calA^\mathsf{u},n,n_s,n_t,\theta)$ runs in time $O(n^2m^3\log(nm)\varepsilon^{-2})$ and returns $\mathsf{est}$ with the guarantee
$
		\Pr\left[\mathsf{est}  \notin (1 \pm \varepsilon) |\calL(\calA^\mathsf{u})| \right] \leq \frac{1}{4}.
$
\end{restatable}
\begin{proof}
	\textcolor{black}{Let $\texttt{A}^{(\ast)} = \approxMCNFAcore^{(\ast)}(\calA^\mathsf{u},n,n_s,n_t,\theta)$.
		\[
		\begin{aligned}
			\Pr_{\texttt{A}}\left[\mathsf{est} \not\in (1 \pm \varepsilon)|\calL_n(\calA)|\right] 
			= \Pr_{\texttt{A}^*}\left[
			\sum_{r,q} |S^r(q)| \geq \theta
			\right] + \Pr_{\texttt{A}^*}\left[
			\sum_{r,q} |S^r(q)| < \theta \text{ and } \mathsf{est} \not\in (1 \pm \varepsilon)|\calL_n(\calA)|
			\right] 
			\\
			\leq \Pr_{\texttt{A}^*}\left[
			\sum_{r,q} |S^r(q)| \geq \theta
			\right] + \Pr_{\texttt{A}^*}\left[
			\mathsf{est} \not\in (1 \pm \varepsilon)|\calL_n(\calA)|
			\right]  
			\leq \Pr_{\texttt{A}^*}\left[
			\sum_{r,q} |S^r(q)| \geq \theta
			\right] + \Pr_{\texttt{A}^*}\left[
			\bigcup_{q}
			p(q) \not\in \frac{(1 \pm \varepsilon)^{-1}}{|\calL(q)|} 
			\right]
		\end{aligned}
		\]
		where $q$ ranges over $\calQ^\mathsf{u}$ and $r$ ranges over $[n_sn_t]$. The number $\kappa$ has been set so that $p(q) \not\in \frac{1 \pm \kappa}{|\calL(q)|}$ implies $p(q) \not\in \frac{(1 \pm \varepsilon)^{-1}}{|\calL(q)|}$ so, using Lemmas~\ref{lemma:proba_p(q)} and~\ref{lemma:proba_S(q)}: 
		\[
		\Pr_{\texttt{A}}\left[\mathsf{est} \not\in (1 \pm \varepsilon)|\calL_n(\calA)|\right] 
		\leq \Pr_{\texttt{A}^*}\left[
		\sum_{r,q} |S^r(q)| \geq \theta
		\right] + \Pr_{\texttt{A}^*}\left[
		\bigcup_{q}
		p(q) \not\in \frac{1 \pm \kappa}{|\calL(q)|} 
		\right] \leq \frac{1}{4} 
		\]
		To analyse the running time, we assume one of the caching strategy is used.} 
	\begin{enumerate}[leftmargin=*]
		\item[•] \textcolor{black}{The cost of $\computeCache{i}$ is twice the cost of multiplying a $|\calS^{i-1}| \times |\calQ^{i-1}|$ matrix with a $|\calQ^{i-1}|\times|\calQ^i|$ matrix, so it is in $O(\max(\frac{|\calS^{i-1}|}{m},1)\cdot m^\omega) = O(\max(|\calS^{i-1}|\cdot  m^{\omega-1}, m^{\omega}))$ where $m = \max_{i} |\calQ^i|$ and $\omega \leq 3$ is the matrix product constant over $\mathbb{N}$. The cost of $\updateCache{i}$ is $O(|\calS^i|\cdot |\calQ^i|)$. So the total cost for caching when the number of samples stays below $\theta$ is at most $O(n\cdot m^{\omega} + \theta\cdot  m^{\omega-1}) = O(\theta\cdot  m^{\omega-1})$ (because $\theta \geq n|\calQ| \geq nm$). We have that $\omega \leq 3$ and $\theta = O(n_sn_t|\calQ^\mathsf{u}|) = O(n_sn_tnm)$ so the cost is in $O(n_sn_tn m^3)$.}
		
		\item[•] \textcolor{black}{For the $\reduce$, each sample $w \in S^r(q)$ with $q \in \calQ^i$ is tested only in $\reduce$ operations triggered by $\estimateAndSample$ when called on states in $\calQ^{i+1}$. So each $w$ can be tested in at most $m$ $\reduce$ operations and therefore the cumulated cost of all $\reduce$ is in $O(\theta  m) = O(n_sn_tn m^2)$.}
		
		\item[•] \textcolor{black}{For the median of means, at most $|\calQ^\mathsf{u}| n_t$ means of $n_s$ integers are computed and at most $|\calQ^\mathsf{u}|$ medians of $n_t$ integers are computed so the cumulated cost of all median of means computations is $O((n_sn_t + n_t\log(n_t))|\calQ^\mathsf{u}|) \leq O(nm(n_sn_t + n_t\log(n_t)))$. $n_t$ is in $O(\log|\calQ^\mathsf{u}|) = O(\log(nm))$ so this sums up to a cost in $O(n_sn_tnm + n_tnm\log\log(nm))$}
		
		\item[•] \textcolor{black}{For the $\union$, with caching it takes $O(m)$-time to determine for $\pred(q) = (q_1,\dots,q_k)$ and $w \in \calL(q_\ell)$ whether $w \in  \calL(q_j)$ for any $j < \ell$. For every a single sample $w$, this query is done only for states $q$ belonging to a same layer $\calQ^i$. So, when the number of samples stays below $\theta$, at most $\theta\cdot m$ queries occur, resulting in a cost in $O(\theta  m^2) = O(n_sn_tn m^3)$.}
	\end{enumerate}
	\textcolor{black}{Thus the overall running time is in $O(n_sn_tnm^3) = O(n^2m^3\varepsilon^{-2}\log(nm))$.}
\end{proof}

\noindent The value $1/4$ is decreased down to any $\delta > 0$ with the median technique, hence our main result.

\mainResult*
\begin{proof}
\textcolor{black}{Let $\mathsf{est}_1,\dots,\mathsf{est}_m$ be the outputs of $n_u$ independent calls to $\approxMCNFAcore$. Let $X_i$ be the indicator variable that takes value $1$ if and only if $\mathsf{est}_i \not\in (1\pm \varepsilon)|\calL_n(\calA)|$, and define $\bar X = X_1 + \dots + X_{n_u}$. By Lemma~\ref{lemma:main_result_core}, $\Ex[\bar X] \leq n_u/4$. Hoeffding bound gives
\[
\Pr\Big[\median(X_1,\dots,X_{n_u}) \not\in (1 \pm \varepsilon)|\calL_n(\calA)|\Big] = \Pr\left[\bar X > \frac{n_u}{2}\right] \leq \Pr\left[\bar X - \Ex[\bar X] > \frac{n_u}{4}\right] \leq e^{-n_u/8} \leq \delta.
\]
The running time is $O(|\log(\frac{1}{\delta})|)$ times that of $\approxMCNFAcore$ (using $n_u \leq |\calQ|$). }
\end{proof}

\subsection{Proof of Lemma~\ref{lemma:proba_p(q)}}

Let $\Delta(q)$ be the interval $\frac{1 \pm \kappa}{|\calL(q)|}$ and $\nabla(q)$ be the interval $\frac{|\calL(q)|}{1 \pm \kappa}$. The probability to bound is  
$$
P_1 := \Pr\left[\bigcup_{q \in \calQ^\mathsf{u}}	p(q) \not\in \Delta(q)\right].
$$ 
The key component for proving Lemma~\ref{lemma:proba_p(q)}, is a variance upper bound on the size of the sample sets. Before that, some work on $P_1$ is needed. The first important idea is that we focus on the probability of the first occurence of $p(q) \not\in \Delta(q)$, that is, when $p(q') \in \Delta(q')$ holds for all ancestors of $q$.

\begin{claim}\label{claim:first_error}
\textcolor{black}{$
p(q) \not\in \Delta(q)
$
occurs for some $q \in \calQ^\mathsf{u}$ if and only if, for some $q' \in \calQ^{> 0}$ we have that 
$p(q') \not\in \Delta(q')$ and for all $q'' \in \ancestors(q')$ we have $p(q'') \in \Delta(q'')$.}
\end{claim}
\begin{proof}
\textcolor{black}{The ``if'' direction is trivial. For the other direction, suppose that $p(q) \not\in \Delta(q)$ holds for some $q$. Let $i$ be the smallest integer such that there is $q \in \calQ^i$ and $p(q) \not\in \Delta(q)$. $i$ cannot be $0$ because the only state in $\calQ^0$ is $q_I$ and $p(q_I) = 1 = |\calL(q_I)|^{-1} \in \Delta(q_I)$. So $q \in \calQ^{> 0}$ and, by minimality of $i$, we have that $p(q'') \in \Delta(q'')$ for all $q'' \in \ancestors(q)$.}
\end{proof}
\textcolor{black}{\begin{align*}
P_1 =
\Pr\left[
		\bigcup_{q \in \calQ^{> 0}}
			p(q) \not\in \Delta(q) 
			\text{ and } \forall q' \in
			 \ancestors(q),\, p(q') \in
			  \Delta(q') 
\right] \tag{Claim~\ref{claim:first_error}}
\\
\leq
\sum_{q \in \calQ^{> 0}}\underbrace{\Pr\left[
			p(q) \not\in \Delta(q) 
			\text{ and } \forall q' \in \ancestors(q),\, p(q') \in
			  \Delta(q')
\right]}_{P_1(q)}
\end{align*}
We introduce the history of $q$ in $P_1(q)$. Consider the set $\calH_q$ of realizable histories for $q$ and denote by $H(q) = h$ the event that the history $h$ occurs for $q$, that is, the event that the algorithm sets $p(q')$ to $h(q')$ for all $q' \in \ancestors(q)$.
\begin{align*}
P_1(q) &= \sum_{h \in \calH_q} \Pr\left[H(q)=h \text{ and }
				p(q) \not\in \Delta(q) 
				\text{ and for all } q' \in
				 \ancestors(q),\, p(q') \in
				  \Delta(q')
		\right].
\end{align*}
The summand probabilities are zero when $h(q')$ is not in $\Delta(q')$ for any ancestor of $q$. Let $\calH^*_q$ be the subset of $\calH_q$ where $h(q') \in \Delta(q')$ holds for every $q' \in \ancestors(q)$. Then
\begin{equation*}
\begin{aligned}
P_1(q) &= \sum_{h \in \calH^*_q} \Pr\left[H(q)=h \text{ and }
				p(q) \not\in \Delta(q) 
				\text{ and for all } q' \in
				 \ancestors(q),\, p(q') \in
				  \Delta(q')
		\right]
\\
	&= \sum_{h \in \calH^*_q} \Pr\left[H(q)=h \text{ and }
				p(q) \not\in \Delta(q)
		\right]
\end{aligned} 
\end{equation*}
Let $R_j$ be the interval $\{n_s(j-1) + 1,\dots,n_sj\}$. Recall that $\rho(q) = \min(p(q_1),\dots,p(q_k))$ and $M^j(q) = (n_s\rho(q))^{-1}\sum_{r \in R_j}|\hat{S}^{r}(q)|$. Let $\rho_h(q) = \min(h(q_1),\dots,h(q_k))$ and $M^j_h(q) = (n_s\rho_h(q))^{-1}\sum_{r \in R_j}|\hat{S}^{r}(q)|$. Note that $\rho_h(q)$ is a constant. Under the event $H(q) = h$, $\estimateAndSample(q)$ sets $p(q)$ to $\min(\rho_h(q),\hat \rho_h(q))$ where $\hat \rho_h(q) = \underset{0 \leq j < n_t}{\median}(M^j_h(q))^{-1}$. We show that it is enough to focus on $\hat \rho_h(q)$.}

\begin{claim}\label{claim:hat_rho}
\textcolor{black}{When $H(q) = h$ occurs, $\hat \rho_h(q) \in \Delta(q)$ implies that $p(q) \in \Delta(q)$.}
\end{claim}
\begin{proof}
\textcolor{black}{If $p(q) = \hat \rho_h(q)$ then, trivially, $\hat \rho_h(q) \in \Delta(q)$ implies $p(q) \in \Delta(q)$. Otherwise if $p(q) = \rho_h(q)$ then $\hat \rho_h(q) \in \Delta(q)$ implies that $p(q) \leq \hat \rho_h(q) \leq (1 + \kappa)/|\calL(q)|$ and, since $h \in \calH^*_q$ guarantees that $\rho_h(q) \geq (1-\kappa)/\max_{j \in [k]} |\calL(q_j)| \geq (1-\kappa)/|\calL(q)|$, we have that $p(q) \in \Delta(q)$.}
\end{proof}
\textcolor{black}{
\begin{align*}
\Pr\left[H(q)=h \text{ and }  p(q) \not\in \Delta(q) \right]
\leq
\Pr\left[H(q)=h \text{ and }  p(q) \not\in \Delta(q) \text{ and } \hat \rho_h(q \not\in \Delta(q) \right] \tag{Claim~\ref{claim:hat_rho}}
\\
=
\Pr\left[H(q)=h \text{ and }  p(q) \not\in \Delta(q) \text{ and } \underset{1 \leq j \leq n_t}{\median}(M^j_h(q)) \not\in \nabla(q) \right] 
\\
\leq 
\Pr\left[H(q)=h \text{ and } \underset{1 \leq j \leq n_t}{\median}(M^j_h(q)) \not\in \nabla(q) \right] 
\end{align*}}

Next, let $\mathfrak{M}^j_h(q) = (\rho_h(q)n_s)^{-1}\sum_{r \in R_j}|\vZ{r}{h}{q}|$ be the counterpart of $M^j_h(q)$ in the random process.  Because $M^j_h(q)$ and $\mathfrak{M}^j_h(q)$ are constructed analogously from $(\hat{S}^{r}(q))_r$ and $(\vZ{r}{h}{q}|)_r$, respectively, Lemma~\ref{lemma:not_so_black_magic} implies that $(H(q), \median_{j \in [n_t]}(M^j_{H(q)}(q)))$ and $(H(q), \median_{j \in [n_t]}(\mathfrak{M}^j_{H(q)}(q)))$ have equal distribution. So we replace $M^j_h(q)$ by $\mathfrak{M}^j_h(q)$ and use the independence of $\mathfrak{M}^j_h$ with $H(q)$ (Fact~\ref{fact:independence}).
\begin{align*}
\Pr\Big[H(q)=h \text{ and }  \underset{0 \leq j < n_t}{\median}(M^j_h(q)) \not\in \nabla(q) \Big]
\leq \Pr[H(q) = h] \Pr\Big[\underset{0 \leq j < n_t}{\median}(\mathfrak{M}^j_h(q)) \not\in \nabla(q)\Big]
\end{align*}
We bound $\Pr[\median_{j \in [n_t]}(\mathfrak{M}^j_h(q)) \not\in \nabla(q)]$
using Chebyshev's inequality followed by Hoeffding bound. For that, we need a bound on the variance of the variables $|\vZ{r}{h}{q}|$. By Lemma~\ref{lemma:proba_first_order}, the expected value of $|\vZ{r}{h}{q}|$ is $\mu := \rho_h(q)|\calL(q)|$. Now for the variance,
\begin{align*}
\Va\left[|\vZ{r}{h}{q}| \right] \leq \Ex\left[|\vZ{r}{h}{q}|^2 \right] = \mu + \sum_{w \in \calL(q)}\sum_{\substack{w' \in \calL(q) \\ w \neq w'}} \Pr\left[w \in \vZ{r}{h}{q} \text{ and } w' \in \vZ{r}{h}{q}\right]
\\
\leq \mu + \sum_{w \in \calL(q)}\sum_{\substack{w' \in \calL(q) \\ w \neq w'}} \frac{\rho_h(q)^2}{h(\lcpn(\paths(q,w),\paths(q,w'))} \tag{Lemma~\ref{lemma:proba_second_order}}
\end{align*}
Let $R = \paths(w,q) = (q^0_w,q^1_w,q^2_w,\dots,q^{i-1}_w,q)$, with $q^0_w = q_F$. Let $R' = \paths(w',q)$ for any $w' \in \calL(q)$ distinct from $w$. Then $\lcpn(R,R')$ is one of the $q^j_w$. Recall that $I(w,q,j)$ is the set of $w' \in \calL(q)$ such that $\lcpn(R,R') = q^j_w$.
\begin{align*}
\sum_{\substack{w' \in \calL(q) \\ w \neq w'}} &\frac{\rho_h^2}{h(\lcpn(\paths(q,w),\paths(q,w'))} = \sum_{j = 0}^{i-1} |I(w,q,j)|\frac{\rho_h(q)^2}{h(q^j_w)} 
\leq \sum_{j = 0}^{i-1} \frac{\rho_h(q)^2 |\calL(q)|}{h(q^j_w)|\calL(q^j_w)|}
\tag{Lemma~\ref{lemma:derivation_path}}
\end{align*}
Since $h \in \calH^*_q$, it holds that $h(q^j_w) \in \Delta(q^j_w)$. This implies $h(q^j_w) \geq (1 - \kappa)/|\calL(q^j_w)|$ and therefore 
$$
\Va\left[|\vZ{r}{h}{q}| \right] \leq
\mu + \sum_{w \in \calL(q)} \sum_{j = 0}^{i-1} |\calL(q)| \leq \mu + \frac{\rho_h(q)^2n}{1 - \kappa}\sum_{w \in \calL(q)} |\calL(q)| = \mu + \frac{\mu^2n}{1 - \kappa}
$$
Now, back to the median of means, we have that $\Ex[\mathfrak{M}^j_h] = \mu/\rho_h(q) = |\calL(q)|$ and, by independence of the variables $\{\vZ{r}{h}{q}\}_{r \in [n_sn_t]}$, the variance is 
\begin{align*}
\Va[\mathfrak{M}^j_h(q)] = \frac{1}{\rho_h(q)^2n_s^2}\sum_{r = j\cdot n_s+1}^{(j+1)n_s} \Va[|\vZ{r}{h}{q}|]\leq \frac{1}{\rho_h(q)^2n_s} \left( \mu + \frac{\mu^2n}{1 - \kappa}\right) =   \frac{1}{n_s}\left(\frac{|\calL(q)|}{\rho_h(q)} + \frac{n |\calL(q)|^2}{1-\kappa}\right).
\end{align*}
Next, $\mathfrak{M}^j_h(q) \in \nabla(q) = \frac{|\calL(q)|}{1 \pm \kappa}$ occurs if and only if $\frac{-\kappa|\calL(q)|}{1 + \kappa} \leq \mathfrak{M}^j_h(q) - |\calL(q)| \leq \frac{\kappa|\calL(q)|}{1 - \kappa}$, which is subsumed by $|\mathfrak{M}^j_h(q) - |\calL(q)|| \leq \frac{\kappa|\calL(q)|}{1+\kappa}$. So Chebyshev's inequality gives 
\begin{align*}
&\Pr\left[\mathfrak{M}^j_h(q) \notin \nabla(q)\right] \leq \Pr\left[\big|\mathfrak{M}^j_h(q) - |\calL(q)|\big| > \frac{\kappa|\calL(q)|}{1+\kappa}\right] \leq \frac{(1+\kappa)^2}{\kappa^2 |\calL(q)|^2} \Va\left[\mathfrak{M}^j_h(q)\right]
\\
&\leq \frac{(1+\kappa)^2}{\kappa^2n_s}\left(\frac{1}{\rho_h(q)|\calL(q)|} + \frac{n}{1-\kappa}\right)
\leq \frac{(1+\kappa)^2}{\kappa^2n_s}\left(\frac{1}{1 - \kappa} + \frac{n}{1-\kappa}\right) \tag{$\rho_h(q) \geq \frac{1- \kappa}{\max_{j}|\calL(q_j)|} \geq \frac{1- \kappa}{|\calL(q)|}$}
\\
&\leq \frac{1}{4} \tag{$n_s \geq \frac{4(n+1)(1+\kappa)^2}{\kappa^2(1-\kappa)}$}
\end{align*}
By taking the median of the $\mathfrak{M}^j_h(q)$'s, we decrease the $1/4$ upper bound. Let $E_j$ be the indicator variable taking value~$1$ if and only if $\mathfrak{M}^j_h(q) \not\in \nabla(q)$ and let $\bar E = \sum_{j = 0}^{n_t - 1} E_j$. We have $\Ex[\bar E] \leq n_t/4$ so Hoeffding bound gives
\begin{align*}
\Pr\left[ \underset{0 \leq j < n_t}{\median}(\mathfrak{M}^j_h(q)) \not\in \nabla(q) \right] 
= 
\Pr\left[ \bar E > \frac{n_t}{2}\right] 
\leq 
\Pr\left[ \bar E - \Ex(\bar E) \geq \frac{n_t}{4}\right] 
\leq e^{-n_t/8} \leq \frac{1}{16|\calQ^\mathsf{u}|}
\end{align*}
where the last inequality comes from $n_t \geq 8\ln(16|\calQ^\mathsf{u}|)$. Putting everything together, we have that 
$$
P_1 \leq \sum_{q \in \calQ^{> 0}}\sum_{h \in \calH^*_q} \frac{1}{16|\calQ^\mathsf{u}|}\Pr[H(q) = h] \leq  \frac{1}{16}.
$$

\subsection*{Proof of Lemma~\ref{lemma:proba_S(q)}}

\textcolor{black}{We reuse the $\Delta(q)$ notation from before.
\begin{align*}
\Pr\left[\sum_{r,q} |S^r(q)| \geq \theta\right] \leq \underbrace{\Pr\left[\sum_{r,q} |S^r(q)| \geq \theta \text{ and } \bigcap_{q' \in \calQ^\mathsf{u}} p(q') \in \Delta(q') \right]}_{P_2} + \underbrace{\Pr\left[\bigcup_{q \in \calQ^\mathsf{u}} p(q) \not\in \Delta(q) \right]}_{P_1}
\end{align*}
Lemma~\ref{lemma:proba_p(q)} gives a $1/16$ upper bound on $P_1$. We focus on $P_2$. We start with Markov's inequality.
\begin{align*}
P_2 &= 
\Pr\left[\sum_{r,q}\bigg(|S^r(q)| \prod_{q' \in \calQ^\mathsf{u}}\mathbbm{1}\left(p(q') \in \Delta(q')\right)\bigg) \geq \theta \right] 
\leq 
\frac{1}{\theta} \cdot \Ex\left[\sum_{r,q}\bigg(|S^r(q)|\prod_{q' \in \calQ^\mathsf{u}}\mathbbm{1}\left(p(q') \in \Delta(q')\right)\bigg)\right] 
\\
&=
\frac{1}{\theta}\cdot \sum_{r,q} \Ex\left[|S^r(q)|\cdot \prod_{q' \in \calQ^\mathsf{u}}\mathbbm{1}\left(p(q') \in \Delta(q')\right)\right] 
\leq
\frac{1}{\theta}\cdot \sum_{r,q} \underbrace{\Ex\left[|S^r(q)|\cdot \mathbbm{1}\left(p(q) \in \Delta(q)\right)\right]}_{E(r,q)}
\end{align*}
To bound $E(r,q)$ we introduce the history of $q$.
\begin{align*}
E(r,q) &= \sum_{t \in \Delta(q)} \Ex\left[|S^r(q)|\cdot \mathbbm{1}(p(q) = t)\right]
= \sum_{h \in \calH_q} \sum_{t \in \Delta(q)} \Ex\left[|S^r(q)|\cdot \mathbbm{1}(p(q) = t)\cdot \mathbbm{1}(H(q) = h)\right]
\\
&= \sum_{h \in \calH_q} \sum_{t \in \Delta(q)} \Ex\left[|\vY{r}{h}{t}{q}|\cdot \mathbbm{1}(p(q) = t)\cdot \mathbbm{1}(H(q) = h)\right] \tag{Lemma~\ref{lemma:not_so_black_magic}}
\\
&= \sum_{h \in \calH_q} \sum_{t \in \Delta(q)} \Ex\left[|\vY{r}{h}{t}{q}| \right]\cdot \Pr[p(q) = t \text{ and }H(q) = h] \tag{Fact~\ref{fact:independence}}
\\
&= \sum_{h \in \calH_q} \sum_{t \in \Delta(q)} t\cdot |\calL(q)|\cdot \Pr[p(q) = t \text{ and } H(q) = h] \tag{Lemma~\ref{lemma:proba_first_order}}
\\
&\leq \sum_{h \in \calH_q} \sum_{t \in \Delta(q)} (1+\kappa)\cdot \Pr[p(q) = t \text{ and } H(q) = h] \leq 1+\kappa \tag{$t \leq (1 + \kappa)/|\calL(q)|$}
\end{align*}
Since $\theta \geq 16(1+\kappa)n_sn_t|\calQ^\mathsf{u}|$, we have $
P_2 \leq (1+\kappa)n_sn_t|\calQ^\mathsf{u}|\theta^{-1} \leq 1/16$ and $\Pr\left[\sum_{r,q} |S^r(q)| \geq \theta\right] \leq P_1 + P_2 \leq 1/8$.}

\section{Conclusion}\label{sec:conclusion}

In this paper, we present a novel FPRAS for \#NFA with a significantly improved time complexity compared to prior schemes. The algorithm has a time complexity that is quadratic in $n$, the word length, and cubic in $m$, the number of states (ignoring logarithmic factors). Notably, the time complexity is sub-quadratic with respect to the membership check's time complexity. Furthermore, we believe that implementations leveraging clever caching schemes could achieve a lower {\em effective} dependence on $m$ in practice. A natural next step would be to pursue algorithmic engineering to develop a practical and scalable \#NFA counter. Naturally, one might wonder whether further improvements to the time complexity are possible. While we do not have formal lower bounds, we are pessimistic about significant improvements. This is because the current time complexity (ignoring dependence on $\varepsilon$ and $\delta$) can be expressed as $\Tilde{O}(nm^2 \cdot nm)$, which corresponds to performing membership checks for different words across all states of the unrolled automaton, given that there are $mn$ states in the unrolled automaton. Thus, we believe reducing this dependence would require fundamentally new ideas. We leave this as a direction for future work.

\begin{acks}
Meel acknowledges the support of the Natural Sciences and Engineering Research Council of Canada (NSERC), [funding reference number RGPIN-2024-05956]; de Colnet is supported by the Austrian Science Fund (FWF),
ESPRIT project FWF ESP 235. This work was done in part while de Colnet was visiting the University of Toronto and Georgia Institute of Technology. This research was initiated at Dagstuhl Seminar 24171 on ``Automated Synthesis: Functional, Reactive and Beyond'' (\url{https://www.dagstuhl.de/24171}). We gratefully acknowledge the Schloss Dagstuhl - Leibniz Center for Informatics for providing an excellent environment and support for scientific collaboration. 
\end{acks}

\appendix

\section*{Appendix}

\subsection*{Proof of Lemmas~\ref{lemma:proba_first_order} and~\ref{lemma:proba_second_order}}

\probaFirstOrder*
\begin{proof}
Let $q \in \calQ^i$. We proceed by induction on $i$. The base case $i = 0$ is immediate since $\vY{r}{h_\emptyset}{1}{q_I} = \{\lambda\} = \calL(q_I)$ and these are the only variables for states in $\calQ^0$. Now let $i > 0$, $q \in \calQ^i$, $\pred(q) = (q_1,\dots,q_k)$ and suppose that $\Pr[w' \in \vY{r}{h'}{t'}{q'}] = t'$ holds for all $\vY{r}{h'}{t'}{q'}$ and $w' \in \calL(q')$ with $q' \in \calQ^{< i}$. Let $t_{\min} = \min(h(q_1),\dots,h(q_k))$ and $w = w' \cdot b$. We have  
\begin{align*}
\Pr\left[w \in \vY{r}{h}{t}{q}\right] 
= \Pr\left[w \in \reduce\left(\vZ{r}{h}{q}, \frac{t}{t_{\min}}\right) \mid w \in Z^r_{q,h}\right]\Pr\left[w \in \vZ{r}{h}{q}\right] 
\\
= \frac{t}{t_{\min}}\Pr\left[w \in \vZ{r}{h}{q}\right]
\end{align*}
Let $t_j = h(q_j)$ and let $h_j$ be the restriction of $h$ to $q_j$'s ancestors. The event $w \in \vZ{r}{h}{q}$ occurs if and only if $w' \in \vY{r}{h_j}{t_j}{q_j}$ for $q_j$ the first $b$-predecessor of $q$ (first with respect to $\prec$) such that $w \in \reduce(\vY{r}{h_j}{t_j}{q_j},\frac{t_{\min}}{t_j})$. So by induction
\begin{align*}
\Pr\left[w \in \vZ{r}{h}{q}\right] = \Pr\left[w \in \reduce\left(\vY{r}{h_j}{t_j}{q_j}, \frac{t_{\min}}{t_j}\right) \mid w' \in \vY{r}{h_j}{t_j}{q_j}\right]\Pr\left[w' \in \vY{r}{h_j}{t_j}{q_j}\right] 
\\
= \frac{t_{\min}}{t_j}\Pr\left[w' \in \vY{r}{h_j}{t_j}{q_j}\right] = t_{\min}
\end{align*}
It follows that $\Pr[w \in \vY{r}{h}{t}{q}] = t$.
\end{proof}

\probaSecondOrder*
\begin{proof}
We prove a stronger statement, namely, that for every $i$, for every $q$ and $q'$ two states (potentially $q = q'$) in $\calQ^i$ and every $t$ and $t'$ such that $t = t'$ when $q = q'$,  and every $w \in \calL(q)$ and $w' \in \calL(q')$, and $h$ and $h'$ two compatible histories for $q$ and $q'$, respectively, we have that 
\begin{equation}\label{equation:stronger_statement}
\Pr\left[w \in \vY{r}{h}{t}{q} \text{ and } w' \in \vY{r}{h'}{t'}{q'}\right] \leq \frac{tt'}{t^*}
\end{equation}
where $t^* = t$ if $q = q'$ and $w = w'$, and $t^* = h(\lcpn(\paths(w,q),\paths(w',q'))$ otherwise.

Inequality (\ref{equation:stronger_statement}) is straightforward when $(q,w) = (q',w')$ because then $t = t' = t^*$ and we use Lemma~\ref{lemma:proba_first_order}. In particular~(\ref{equation:stronger_statement}) holds true when $q = q' = q_I$. Now we assume $(w, q) \neq (w',q')$ and proceed by induction on $i$. The base case $i = 0$ holds true by the previous remark. 

Let $b,b' \in \{0,1\}$ such that $w = \omega \cdot b$ and $w' = \omega' \cdot b'$ (potentially $\omega = \omega' = \lambda$). Let $t_{\min}$ (resp. $t'_{\min}$) be the minimum $h(q'')$ (resp. $h'(q'')$) for $q'' \in \pred(q)$ (resp. $q'' \in  \pred(q')$). We have that 
\begin{align*}
&\Pr\left[w \in \vY{r}{h}{t}{q} \text{ and } w' \in \vY{r}{h'}{t'}{q'}\right] = \Pr\left[w \in \vZ{r}{h}{q} \text{ and } w' \in \vZ{r}{h'}{q'}\right]\\
 \times&\Pr\left[w \in \reduce\left(\vZ{r}{h}{q},\frac{t}{t_{\min}}\right) \text{, } w' \in \reduce\left(\vZ{r}{h'}{q'},\frac{t'}{t'_{\min}}\right) \,\Big|\, w \in \vZ{r}{h}{q} \text{, } w' \in \vZ{r}{h'}{q'}\right] 
\end{align*}
We have that 
$\Pr[w \in \reduce(\vZ{r}{h}{q},\frac{t}{t_{\min}}) \text{ and } w' \in \reduce (\vZ{r}{h'}{q'},\frac{t'}{t'_{\min}}) \mid w \in \vZ{r}{h}{q} \text{ and }  w' \in \vZ{r}{h'}{q'}]$ equals $\Pr[w \in \reduce(\vZ{r}{h}{q},\frac{t}{t_{\min}}) \mid w \in \vZ{r}{h}{q}]\Pr[ w' \in \reduce (\vZ{r}{h'}{q'},\frac{t'}{t'_{\min}}) \mid w' \in \vZ{r}{h'}{q'}] = \frac{tt'}{t_{\min}t'_{\min}}$ by the independence of the $\reduce$ events. So
\begin{equation}\label{equation:vZ_connection}
\Pr\left[w \in \vY{r}{h}{t}{q} \text{ and } w' \in \vY{r}{h'}{t'}{q'}\right] = \frac{tt'}{t_{\min}t'_{\min}} \Pr\left[w \in \vZ{r}{h}{q} \text{ and } w' \in \vZ{r}{h'}{q'}\right]
\end{equation}
There is a unique $b\textit{-}$predecessor $q_j$ of $q$ and a unique $b'$-predecessor $q'_l$ of $q'$ such that $\Pr[w \in \vZ{r}{h}{q} \text{ and } w' \in \vZ{r}{h'}{q'}]$ equals
\begin{align*}
&\Pr\left[\omega \in \reduce\left(\vY{r}{h_j}{t_j}{q_j},\frac{t_{\min}}{t_j}\right), \omega' \in \reduce\left(\vY{r}{h'_l}{t'_l}{q'_l},\frac{t'_{\min}}{t'_l}\right) \,\Big|\, \omega \in \vY{r}{h_j}{t_j}{q_j} \text{, } \omega' \in \vY{r}{h'_l}{t'_l}{q'_l}\right] 
\\ \times & \Pr\left[\omega \in \vY{r}{h_j}{t_j}{q_j}  \text{ and } \omega' \in \vY{r}{h'_l}{t'_l}{q'_l}\right]
\end{align*}
where $h_j$ denotes the restriction of $h$ to $\ancestors(q_j)$, $h'_l$ denotes the restriction of $h'$ to $\ancestors(q'_l)$, $t_j = h(q_j)$ and $t'_l = h'(q'_l)$.
By independence of the $\reduce$ events, the first probability equals $\frac{t_{\min}t'_{\min}}{t_jt'_l}$. The states $q_j$ and $q'_l$ are in $\calQ^{i-1}$ so, by induction, $\Pr[\omega \in \vY{r}{h_j}{t_j}{q_j}  \text{ and } \omega' \in \vY{r}{h'_l}{t'_l}{q'_l}]$ is bounded from above by $\frac{t_jt'_l}{\tau}$ where $\tau = t_j$ if $(q_j,t_j,\omega) = (q'_l,t'_l,\omega')$ and $\tau = h(\lcpn(\paths(\omega,q_j),(\omega',q'_l)))$ otherwise. So 
$$
\Pr[w \in \vZ{r}{h}{q} \text{ and } w' \in \vZ{r}{h'}{q'}] \leq \frac{t_{\min}t'_{\min}}{\tau}.
$$
If $(q_j,t_j,\omega) = (q'_l,t'_l,\omega')$ then $q \neq q'$ for otherwise we would have $b = b'$ and $(q,w) = (q',w')$. So $\lcpn(\paths(w,q),\paths(w',q')) = q_j = q'_l$ and $t^* = h(q_j) = h(q'_l) = t'_l = t_j = \tau$. If instead $(q_j,t_j,\omega) \neq (q'_l,t'_l,\omega')$ then $\lcpn(\paths(w,q),\paths(w',q')) = \lcpn(\paths(\omega,q_j),\paths(\omega',q'_l))$ and, again $t^* = \tau$. It follows that $\Pr[w \in \vZ{r}{h}{q} \text{ and } w' \in \vZ{r}{h'}{q'}] \leq \frac{t_{\min}t'_{\min}}{t^*}$ and thus $\Pr[w \in \vY{r}{h}{t}{q} \text{ and } w' \in \vY{r}{h'}{t'}{q'}] \leq \frac{tt'}{t^*}$. The matching inequality for the $\vZ{r}{h}{q}$ variables follows from (\ref{equation:vZ_connection}).
\end{proof}

\notSoBlackMagic*

We make a sequence of modifications to $\approxMCNFAcore^*$ to transform it into the random process. The modifications are on $\estimateAndSample$. Initially the procedure is as follows:

\renewcommand{\thealgocf}{}

\begin{minipage}{\textwidth}
\begin{algorithm}[H]
compute $\hat{S}^r(q)$ for all $r$ using $\{ S^r(q') \mid q' \in \pred(q)\}$\\
compute $p(q)$ using $\{\hat{S}^r(q)\}_r$\\
compute $S^r(q)$ for all $r$ using $p(q)$ and $\hat{S}^r(q)$
\caption*{$\estimateAndSample(q)$}
\end{algorithm}
\end{minipage}

\paragraph*{Recording histories.} For every state $q$, we keep the history for the ancestors of $q$. Formally, we maintain a variable $H(q)$ for every state $q$. Initially $H(q)$ is empty for all $q$. After the value $p(q)$ is computed the variable $H(q')$ is updated for all descendants $q'$ of $q$ as follow: $H(q') = H(q') \cup (p \mapsto p(q))$. Thus, just before processing $q'$, $H(q')$ is the history for $q'$. Clearly, keeping unused additional variables does not modify the output of the algorithm. The procedure $\estimateAndSample$ now is as follow:
 
\begin{minipage}{\textwidth}
\begin{algorithm}[H]
compute $\hat{S}^r(q)$ for all $r$ using $\{ S^r(q') \mid q' \in \pred(q)\}$\\
compute $p(q)$ using $\{\hat{S}^r(q)\}_r$ \textcolor{red}{and update $H$}\\
compute $S^r(q)$ for all $r$ using $p(q)$ and $\hat{S}^r(q)$
\caption*{$\estimateAndSample(q)$}
\end{algorithm}
\end{minipage}

\paragraph*{Introducing clone variables.} For every realizable history $h$ for $q$, and every $r \in [n_sn_t]$, the algorithm now has a variable $\hat{S}^r_h(q)$ and for every $t$ that is a possible candidate for $p(q)$, it also has a variable $S^r_{h,t}(q)$. Initially all $\hat{S}^r_{h}(q)$ and $S^r_{h,t}(q)$ are empty. When the algorithm computes the value for $p(q)$, $H(q)$ has already been set. Once $\hat{S}^r(q)$ is computed we copy its content in $\hat{S}^r_{H(q)}(q)$ and once $S^r(q)$ is computed we copy its content in $S^r_{H(q),p(q)}(q)$. All these steps are superfluous since we are just assigning variables that, so far, are not used. Note that
\begin{align*}
(H(q),\hat S^r(q)) = (H(q),\hat S^r_{H(q)}(q)) \quad  \text{ and } \quad (H(q),p(q), S^r(q)) = (H(q),p(q),S^r_{H(q),p(q)}(q)).
\end{align*}
 
\begin{minipage}{\textwidth}
\begin{algorithm}[H]
compute $\hat{S}^r(q)$ for all $r$ using $\{ S^r(q') \mid q' \in \pred(q)\}$\\
\textcolor{red}{copy $\hat{S}^r(q)$ to $\hat{S}^r_{H(q)}(q)$ for all $r$}\\
compute $p(q)$ using $\{\hat{S}^r(q)\}_r$ and update $H$\\
compute $S^r(q)$ for all $r$ using $p(q)$ and $\hat{S}^r(q)$\\
\textcolor{red}{copy $S^r(q)$ to $S^r_{H(q),p(q)}(q)$ for all $r$}
\caption*{$\estimateAndSample(q)$}
\end{algorithm}
\end{minipage}

\paragraph*{Working with the clone variables.}
For a fixed $q$, to compute $p(q)$, $\{\hat S^r(q)\}_r$ and $\{S^r(q)\}_r$, the algorithm uses $\{p(q')\}_{q' \in \pred(q)}$ and $\{\{S^r(q')\}_r\}_{q' \in \pred(q)}$. By the equalities above, we can use the $S^r_{h,t}(q)$ instead of the $S^r(q)$ to do the computation.  We just need to retrieve the correct sets using $H(q)$.

\begin{minipage}{\textwidth}
\begin{algorithm}[H]
compute $\hat{S}^r(q)$ for all $r$ using $\{\textcolor{red}{S^r_{H(q'),p(q')}(q')} \mid q' \in \pred(q)\}$\\
copy $\hat{S}^r(q)$ to $\hat{S}^r_{H(q)}(q)$ for all $r$\\
compute $p(q)$ using $\{\textcolor{red}{\hat{S}^r_{H(q)}(q)}\}_r$ and update $H$\\
compute $S^r(q)$ for all $r$ using $p(q)$ and \textcolor{red}{$\hat{S}_{H(q)}^r(q)$}\\
copy $S^r(q)$ to $S^r_{H(q),p(q)}(q)$ for all $r$
\caption*{$\estimateAndSample(q)$}
\end{algorithm}
\end{minipage}
\medskip

\noindent But then we might as well compute $\hat{S}^r_{H(q)}(q)$ and $S^r_{H(q),p(q)}(q)$ directly and then copy their content to $S^r(q)$ and $\hat{S}^r(q)$.

\begin{minipage}{\textwidth}
\begin{algorithm}[H]
compute \textcolor{red}{$\hat{S}^r_{H(q)}(q)$} for all $r$ using $\{S^r_{H(q'),p(q')}(q') \mid q' \in \pred(q)\}$\\
compute $p(q)$ using $\{\hat{S}^r_{H(q)}(q)\}_r$ and update $H$\\
compute \textcolor{red}{$S^r_{H(q),p(q)}(q)$} for all $r$ using $\hat{S}^r_{H(q)}(q)$\\
copy \textcolor{red}{$S^r_{H(q),p(q)}(q)$ to $S^r(q)$} for all $r$\\
copy \textcolor{red}{$\hat{S}^r_{H(q)}(q)$ to $\hat{S}^r(q)$} for all $r$
\caption*{$\estimateAndSample(q)$}
\end{algorithm}
\end{minipage}
\medskip

\noindent We still have $(H(q),\hat S^r(q)) = (H(q),\hat S^r_{H(q)}(q))$ and $(H(q),p(q), S^r(q)) = (H(q),p(q),S^r_{H(q),p(q)}(q))$.

\paragraph*{Filling unused clone variables.} When processing a state $q$ given $H(q)$ we know the sets $\hat{S}^r_{h}(q)$ for $h \neq H(q)$ will not be filled and will not be used to compute the output of the algorithm. Similarly once $p(q)$ is found, we know the sets $S^r_{h,t}(q)$ for $(h,t) \neq (H(q),p(q))$ will not be filled nor used to compute the output of the algorithm. So additional work can be done to decide the content of these $\hat{S}^r_{h}(q)$ and $S^r_{h,t}(q)$. In particular, if $\pred(q) = (q_1,\dots,q_k)$ then we can set $\rho_h(q) = \min(h(q_1),\dots,h(q_k)$ and $\hat{S}^r_{h}(q) = \union(q,\reduce(S^r_{h_1,h(q_1)}(q_1),\rho^{h}/h(q_1)),\dots,\reduce(S^r_{h_k,h(q_k)}(q_k),\rho^{h}/h(q_k)))$ and $S^r_{h,t}(q) = \reduce(\hat{S}^r_{h}(q),t/\rho)$, where $h_i$ denote the restriction of $h$ to the ancestors of $q_i$.

\begin{minipage}{0.9\textwidth}
\begin{algorithm}[H]
compute $\hat{S}^r_{H(q)}(q)$ for all $r$ using $\{ S^r_{H(q'),p(q')}(q') \mid q' \in \pred(q)\}$\\
\textcolor{red}{compute $\hat{S}^r_{h}(q)$ for all $r$ and $h$ using $\{ \{S^r_{h,t}(q') \}_{r,h,t} \mid q' \in \pred(q)\}$}\\
compute $p(q)$ using $\{\hat{S}^r_{H(q)}(q)\}_r$ and update $H$\\
compute $S^r_{H(q),p(q)}(q)$ for all $r$ using $\hat{S}^r_{H(q)}(q)$\\
\textcolor{red}{compute $S^r_{h,t}(q)$ for all $r$ and $h \neq H(q)$ using $\{\hat{S}^r_{h}(q)\}_h$}\\
copy $S^r_{H(q),p(q)}(q)$ to $S^r(q)$ for all $r$\\
copy $\hat{S}^r_{H(q)}(q)$ to $\hat{S}^r(q)$ for all $r$
\caption*{$\estimateAndSample(q)$}
\end{algorithm}
\end{minipage}
\medskip

\noindent But then the computation of $\hat S^r_h$ and $S^r_{h,t}$ is no different from what is done to compute $\hat{S}^r_{H(q)}(q)$ and $S^{r}_{H(q),p(q)}(q)$ so the procedure is equivalent to 

\begin{minipage}{\textwidth}
\begin{algorithm}[H]
compute $\hat{S}^r_{h}(q)$ for all $r$ \textcolor{red}{and all $h$ (including $h= H(q)$)} using $\{ S^r_{h',p(q')}(q')\}_{h',q' \in \pred(q)}$\\
compute $p(q)$ using $\{\hat{S}^r_{H(q)}(q)\}_r$ and update $H$\\
compute $S^r_{h,t}(q)$ for all $r$, \textcolor{red}{and all $h$ and $t$ (including $h = H(q)$ and $t = p(q)$)} using $\{\hat{S}^r_{h}(q)\}_h$\\
copy $S^r_{H(q),p(q)}(q)$ to $S^r(q)$ for all $r$\\
copy $\hat{S}^r_{H(q)}(q)$ to $\hat{S}^r(q)$ for all $r$
\caption*{$\estimateAndSample(q)$}
\end{algorithm}
\end{minipage}
\medskip

which is equivalent to

\begin{minipage}{\textwidth}
\begin{algorithm}[H]
compute $\hat{S}^r_{h}(q)$ for all $r$ and all $h$ (including $h= H(q)$) using $\{ S^r_{h',p(q')}(q')\}_{h',q' \in \pred(q)}$\\
compute $S^r_{h,t}(q)$ for all $r$, $h$ and $t$ (including $h = H(q)$ and $t = p(q)$) using $\{\hat{S}^r_{h}(q)\}_h$\\
compute $p(q)$ using $\{\hat{S}^r_{H(q)}(q)\}_r$ and update $H$\\
copy $S^r_{H(q),p(q)}(q)$ to $S^r(q)$ for all $r$\\
copy $\hat{S}^r_{H(q)}(q)$ to $\hat{S}^r(q)$ for all $r$
\caption*{$\estimateAndSample(q)$}
\end{algorithm}
\end{minipage}
\medskip

Now the computation of $p(q)$ and $H$ is not needed to compute $\hat{S}^r_{h}(q)$ and $S^r_{h,t}(q)$. So we have transformed the algorithm into the random process with $\hat{S}^r_h(q) =  \vZ{r}{h}{q}$ and $S^r_{h,t}(q) = \vY{r}{h}{t}{q}$ (the first two lines of $\estimateAndSample(q)$). The variables $H(q)$ and $p(q)$ do the bookkeeping and allow us to retrieve $S^r(q)$ as $S^r_{H(q),p(q)}(q)$ and $\hat S^r(q)$ has $\hat S^r_{H(q)}(q)$ for every $q$. So we indeed have that 
$$
(H(q),\hat S^r(q)) = (H(q),\hat S^r_{H(q)}(q)) = (H(q),\vZ{r}{H(q)}{q})
$$
and
$$
(H(q),p(q),S^r(q)) = (H(q),p(q),S^r_{H(q),p(q)}(q)) = (H(q),p(q),\vY{r}{H(q)}{p(q)}{q}).
$$

\end{document}